\documentclass[twocolumn,aps,superscriptaddress,preprintnumbers,floatfix]{revtex4-2}

\usepackage[utf8]{inputenc}
\usepackage[colorlinks=true, allcolors=blue]{hyperref}
\usepackage{amsfonts}
\usepackage{amsmath}
\usepackage{amsthm}
\usepackage{amssymb}
\usepackage{amsthm}
\usepackage{dsfont}
\usepackage{enumitem}
\usepackage{tikz}
\usepackage{anyfontsize}
\usepackage{array}
\usepackage{natbib}
\usepackage{soul}

\newcommand{\la}[1]{\langle{#1}|}
\newcommand{\ra}[1]{|{#1}\rangle}
\newcommand{\avg}[1]{\langle{#1}\rangle}

\newcommand{\kt}{k_{\text{B}}T}

\newtheorem{defi}{Definition}[section]
\newtheorem{theorem}{Theorem}[section]

\usepackage{subfigure}

\begin{document}

\title{Efficient Quantum Work Reservoirs at the Nanoscale}

\author{Jinghao Lyu}
\email{jolyu@ucdavis.edu}
\affiliation{Complexity Sciences Center and Department of Physics and Astronomy, University of California, Davis, One Shields Avenue, Davis, CA 95616}

\author{Alexander B. Boyd}
\email{alboyd@tcd.ie}
\affiliation{School of Physics, Trinity College Dublin,  College Green, Dublin 2, Ireland}

\author{James P. Crutchfield}
\email{chaos@ucdavis.edu}
\affiliation{Complexity Sciences Center and Department of Physics and Astronomy, University of California, Davis, One Shields Avenue, Davis, CA 95616}


\bibliographystyle{unsrt}

\begin{abstract}
When reformulated as a resource theory, thermodynamics can analyze system behaviors in the single-shot regime. In this, the work required to implement state transitions is bounded by $\alpha-$R\'enyi divergences and so differs in identifying efficient operations compared to stochastic thermodynamics. Thus, a detailed understanding of the difference between stochastic and resource-theoretic thermodynamics is needed. To this end, we explore reversibility in the single-shot regime, generalizing the two-level work reservoirs used there to multi-level work reservoirs. This achieves reversibility in any transition in the single-shot regime. Building on this, we systematically develop multi-level work reservoirs in the nondissipation regime with and without catalysts. The resource-theoretic results show that two-level work reservoirs undershoot Landauer's bound, misleadingly implying energy dissipation during computation. In contrast, we demonstrate that multi-level work reservoirs achieve Landauer's bound while producing arbitrarily low entropy.
\end{abstract}

\preprint{arxiv.org:2305.17815}

\date{\today}


\maketitle

\section{Introduction}

The Second Law of thermodynamics states that the total entropy of a system and its surrounding environment increases when undergoing a transformation---the entropy production of any thermodynamic transformation is nonnegative \cite{esposito2010entropy}.  
This places strong resource bounds on computations performed by a Hamiltonian system coupled to a single thermal bath at temperature $T$. Specifically, the work that can be extracted in transforming a system between potentially nonequilibrium states (from $\rho$ to $\rho'$) is bounded above by the reduction in nonequilibrium free energy \cite{parrondo2015thermodynamics,alicki2004thermodynamics,sagawa2013unitary}:
\begin{align}
\avg{W}_{\mathrm{max}} & = F(\rho) - F(\rho') \nonumber \\
    & = k_{B}T \left[D_{1}(\rho||\tau)-D_{1}(\rho'||\tau)\right]
    ~.
\label{eqn:regularwork}
\end{align}
Here, $k_{B}$ is Boltzmann's constant, $F(\rho) = \mathrm{Tr}(\rho H) - T S(\rho)$ is the nonequilibrium free energy with $S(\rho)\equiv -\text{Tr}\left[ \rho \log \rho \right]$ the von Neumann entropy, $D_{1}(\rho||\tau)\equiv \text{Tr}\left[  \rho \ln {\rho}-\rho \ln{\tau}\right]$ is the relative entropy between $\rho$ and $\tau$, and $\tau$ the Gibbs state with Hamiltonian $H$. This result is a general expression of Landauer's principle, which relates information processing to the energy requirements for a computation \cite{landauer1961irreversibility}.

From the perspective of thermodynamic control, we can achieve Landauer's bound on work \cite{jun2014high} by evolving the system under a time-dependent Hamiltonian $H_S(t)$, while maintaining weak coupling to a thermal reservoir \cite{jarzynski2000hamiltonian}. However, the resulting unitary operator from this Hamiltonian control does not necessarily preserve the total energy of the thermal bath and the system. Rather, the extracted work is the negative total energy difference of the system and bath together \cite{deffner2013information}. Stochastic thermodynamics addresses work production as the result of external control, without explicitly describing the battery that stores the harvested work energy. This begs the question: What are the thermodynamic limits when accounting for the dynamics of the battery that drives a state transition forward? 
This requires a more detailed accounting of resources.

Recently, thermodynamics was reformulated as a resource theory---alternately called single-shot thermodynamics, resource theory of athermality, or simply nanoscale thermodynamics \cite{janzing2000thermodynamic, dahlsten2011inadequacy, brandao2013resource, horodecki2013fundamental, brandao2015second}. In resource theory, work must be stored in specific subsystems that we refer to as \emph{work reservoirs} and function as batteries to power state transitions.  In parallel to thermal reservoirs, a work reservoir is defined by a specific relationship between its energy and entropy: a change in energy corresponds to zero entropy change.  External control cannot violate energy conservation. That is, the unitary evolution of bath, system, and work reservoir together must commute with the joint free Hamiltonian.

Typically, a work reservoir is a two-level quantum system and the corresponding work is called \emph{deterministic work} \cite{horodecki2013fundamental}. The work reservoir starts in one pure state at the beginning and ends in another pure state. The work is defined as the energy gap between those two levels. The deterministic work that can be extracted from the state transition $\rho \to \tau$ is \cite{horodecki2013fundamental}: 
\begin{align}
    W_{\mathrm{one-shot}}^{\mathrm{ext}}=k_{B}T D_{0}(\rho || \tau)
    ~,
\label{eq:RT_OneShot}
\end{align}
where $D_{\alpha}(\rho||\tau) \equiv \frac{1}{\alpha-1}\log \text{Tr} \left[ \rho^\alpha \tau^{1-\alpha} \right]$ is the R\'enyi $\alpha-$divergence between state $\rho$ and $\tau$ \cite{petz1986quasi}. 

This work extraction result differs from the bound set by the Second Law of thermodynamics in Eq. \eqref{eqn:regularwork}, which would yield the result $D_{1}(\rho||\tau)$. Recall that $\alpha=0$ R\'enyi divergence vanishes when both $\rho$ and $\sigma$ have full rank. So in the deterministic work setup, if we have a full rank state $\rho$, there is no work we can extract from it.

However, there is a connection between these two work values. The thermodynamic bound is recovered by considering many copies of $\rho$ and tolerating error $\epsilon$.  If we loosen the requirement such that the final state can be $\epsilon$-close to the copies of thermal states, the work can be described by the smoothed version of $\alpha=0$ R\'enyi divergence \cite{tomamichel2009fully,horodecki2013fundamental}:
\begin{align}
    \lim_{\epsilon \to 0}\lim_{n \to \infty} \frac{1}{n} D_{0}^{(\epsilon)}(\rho^{\otimes n}||\tau^{\otimes n}) = D_{1}(\rho||\tau)~.
\end{align}
We expect this since the classical thermodynamic result is supposed to be correct for  a large ensemble of
identical systems. Since the R\'enyi divergence is nondecreasing as a function of order $\alpha$ \cite{van2014renyi}, we have:
\begin{align}
W_{\mathrm{one-shot}}^{\mathrm{ext}} & = \kt D_{0}(\rho\|\tau) \nonumber\\
   & \leq \kt D_{1}(\rho\|\tau)
  ~.
\end{align}
That is, the resource-theoretic bound on work extractable from state $\rho$ is tighter than Landauer's bound of stochastic thermodynamics.

The two-level constraint also leads to tighter bounds in state formation. 
The deterministic work to form system state $\rho$ in single-shot thermodynamics is \cite{horodecki2013fundamental}:
\begin{align}
    W^{\mathrm{form}}_{\mathrm{one-shot}}= -\kt D_{\infty}(\rho||\tau)
    ~.
\label{eq:RT_OneShot_TwoLevel}
\end{align}
Here, the minus sign indicates that work must be supplied to form the state $\rho$. Similar to extraction, one-shot analysis puts a tighter bound on state formation than Landauer's bound:
\begin{align}
    W^\text{form}_\text{one-shot} & = -k_BT D_\infty(\rho||\tau) \nonumber \\
    & \leq -k_B T D_1 (\rho||\tau)
    ~.
\end{align}

In some cases, $W^{\mathrm{form}}_{\mathrm{one-shot}}$ and $W^{\mathrm{ext}}_{\mathrm{one-shot}}$ equal the average results from thermodynamics. Landauer's bound on erasure \cite{landauer1961irreversibility} and the energy that can be stored in a work reservoir by randomizing a pure bit are both $\kt\log2$ \cite{renes2014work}. However,  resource-theoretic results, such as in Eqs. (\ref{eq:RT_OneShot}) and (\ref{eq:RT_OneShot_TwoLevel}) with two-level work reservoirs, usually undershoot Landauer's bound \cite{chubb2018beyond}.  Energy must be dissipated during state transitions \cite{alhambra2016fluctuating,halpern2015introducing,dahlsten2017entropic,halpern2018maximum,biswas2022fluctuation,taranto2023landauer}.

The following establishes that the disparity arises from assuming that work is stored in a two-level system. We show how to approach the thermodynamic limit of Landauer's bound in nanoscale thermodynamics by abandoning two-level work reservoirs.  When using multi-level work reservoirs as shown in Fig. \ref{fig:generalpic}, thermodynamically efficient state transformations are directly implementable.

\begin{figure}[t]
    \centering
    \includegraphics[width=\columnwidth]{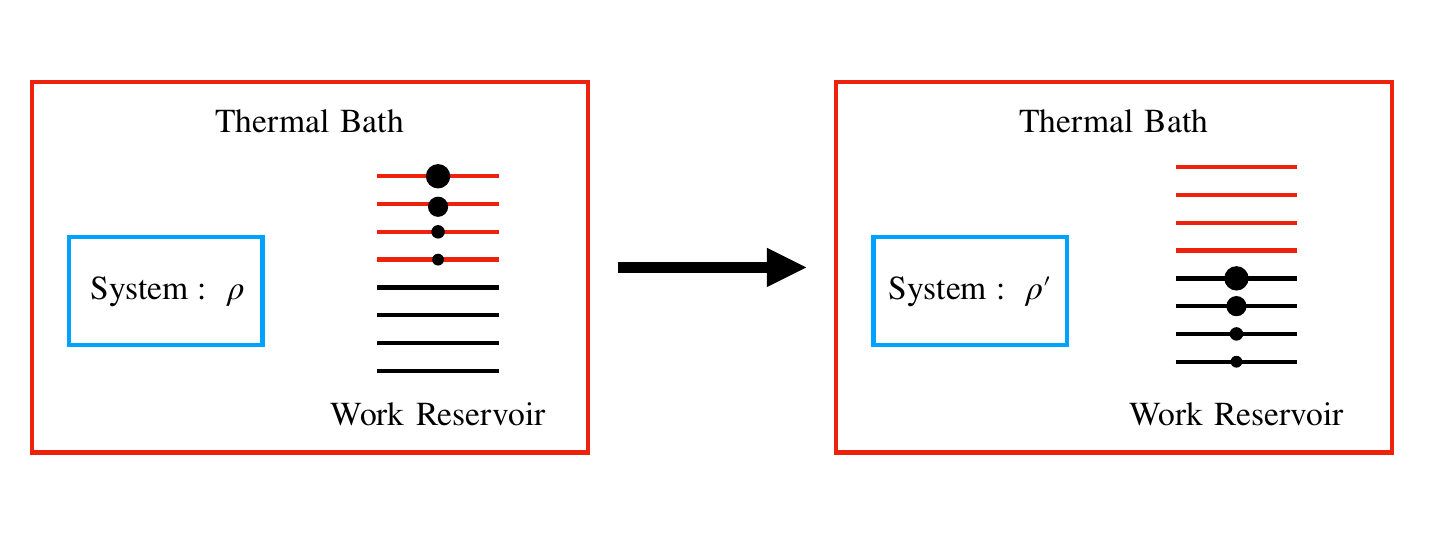}
    \caption{State transitions with multi-level work reservoirs rather than two-level work reservoirs. We show that for any transition $\rho \to \rho'$, there is a multi-level work reservoir such that the dissipation can be arbitrarily small.}
\label{fig:generalpic}
\end{figure}

Our development is organized as follows. Section \ref{section:framework} sets up the basic framework. Section \ref{section:entropyproduction} reviews the definition of entropy production at both the macroscopic scale and the nanoscale and gives an equivalent condition of approaching zero dissipation at the nanoscale. Section \ref{section:beyonddeterministicwork} generalizes the two-level work reservoirs typically employed in nanoscale thermodynamics. It gives an explicit construction for a multi-level work reservoir that can be used to approach zero entropy production for any state transition. Section \ref{section:efficientworkreservoirwithcatalysts} goes on to study efficient work reservoirs in the presence of catalysts and introduces an alternative way to describe almost-nondissipation scenarios.

\section{Framework}\label{section:framework}


The total system consists of system $S$, work reservoir $W$, and thermal bath $B$ with Hamiltonians $H_{S}$, $H_{W}$, and $H_{B}$, respectively. Initially, they are uncorrelated. The initial state is $\rho_{SWB} = \rho_{S} \otimes \rho_{W} \otimes \tau_{B}$, where $\tau_{B}$ is the Gibbs state of the thermal bath at temperature $T$. The three subsystems interact via Hamiltonian $H_{int}$. They evolve by the unitary operator $ U=\mathrm {T} \exp \left({-{\frac {i}{\hbar }}\int H dt}\right)$, where $\mathrm{T}$ is the time-ordering operator and $H$ is the total Hamiltonian $H=H_{S}+H_{B}+H_{W}+H_{int}$. In thermodynamics, there is often no need to include a work reservoir and $U$ does not preserve total energy in general. In resource theory, though, we specify that $[U, H_{S}+H_{B}+H_{W}]=0$---strict energy conservation. The final state is given by $\rho'_{SWB}=U\rho_{SWB} U^{\dagger}$.

Here, we focus on states that are incoherent in energy. Since incoherent states are diagonal in the energy eigenstates, we identify a quantum state $\rho$ with the vector $\boldsymbol{p}$ of its eigenvalues, a Hamiltonian $H$ with its energy levels $\boldsymbol{E}$, and the eigenstates of Hamiltonian $H$ with a classical set $\mathcal{S}=\{1,2, \cdots \}$. Throughout, greek letter $\rho$ denotes a state, bold $\boldsymbol{p}$ denotes a probability distribution, and $p_{i}$/$(\boldsymbol{p})_{i}$ the $i$-th component in the latter. $\tau$ denotes the Gibbs state and $\boldsymbol{\tau}$ the corresponding distribution. Subscripted notation $(\cdot)_{S}$ or $(\cdot)_{B}$ refers to the system or the thermal bath, respectively, while $(\cdot)_{SW}$ denotes the joint distribution of the system and the work reservoir. Notation without subscripts refers to a general state. Primed notation $(\cdot)'$ refers to a final state.



\section{Approach to zero entropy production and work bounds}
\label{section:entropyproduction}

This section reviews the bounds mentioned above and entropy production in single-shot thermodynamics. 

Thermodynamic entropy production $\Sigma$ is defined as \cite{deffner2011nonequilibrium,landi2021irreversible}: 
\begin{align}
    \Sigma = \Delta S_{S} +  \frac{Q}{T}
    ~,
\label{eqn:entropyproduction}
\end{align}
where $\Delta S_{S}$ is the system's entropy change and $Q$ is the amount of heat transferred from the system to the thermal bath. 

We assume that the system and bath are initially uncorrelated and the bath is in equilibrium, such that $\rho_{SB}= \rho_{S} \otimes \tau_{B}$. The global unitary operator $U$ acts on the system and bath to extract work. Using Klein's inequality---$\mathrm{Tr}(\rho \log \rho)\geq \mathrm{Tr}(\rho \log \sigma)$---we can show that the entropy production $\Sigma$ is nonnegative \cite{sagawa2012thermodynamics}.
Define the missing energy of the total system as work production $W= -Q - \Delta U_{S}$, where $\Delta U_{S}$ is the system's energy change, and rewrite Eq. \eqref{eqn:entropyproduction} as:
\begin{align}
    \Sigma = \frac{1}{T}(T \Delta S_{S} - \Delta U_{S} - W) \geq 0
    ~.
\end{align}
This gives the familiar's thermodynamic bound $W\leq -\Delta F_{S}$, where $\Delta F_{S}= \Delta U_{S} -T \Delta S_{S}$. The equal sign holds if and only if the entropy production vanishes. 

Resource theory limits thermodynamic evolution to unitary operators $U$ that commute with the total free Hamiltonian. So, there can be no ``missing energy'': $-Q-\Delta U_{S}=0$. Such operations on the system are called \emph{thermal operations} (TO). 

Without work input, the constraint on state transitions is \emph{thermomajorization} \cite{horodecki2013fundamental, ruch1978mixing}. That is, to transition from $\rho_{S}$ to $\rho_{S}'$ , $\rho_{S}$ must \emph{thermomajorize} $\rho_{S}'$. There is a geometric way to determine this condition: thermomajorization \emph{curves} reveal whether a state $\rho_{S}$ thermomajorizes $\rho_{S}'$ \cite{horodecki2013fundamental}. 

For any state $\rho$, the thermomajorization curve is constructed as follows. Suppose the eigenvalues of $\rho$ are $\boldsymbol{p}=\{{p}_{i}\}_{i\in \mathcal{S}}$ and the corresponding energy levels are $\boldsymbol{E}=\{e_{i}\}_{i\in \mathcal{S}}$. We first rank $\{{p}_{i}\}_{i=1}^{n}$ in descending order of ${p}_{i}e^{\beta e_{i}}$. This is called \emph{$\beta-$order}. The thermomajorization curve of state $\rho$ is formed by connecting points:
\begin{align}
    (0,0) ~\text{and}~
    \Big( \sum_{i=1}^{k} e^{-\beta {e}_{i}^{\downarrow}}, \sum_{i=1}^{k} {p}_{i}^{\downarrow} \Big)_{k=1}^{n}
\end{align}
piecewise linearly where $\downarrow$ means that ${p}_{i}$ and ${e}_{i}$ have been $\beta-$ordered. The thermomajorization curve of state $\rho$ is a monotonic concave-down curve $f_{\boldsymbol{p},\boldsymbol{E}}(x)$ that interpolates between $(x, f(x))= (0,0)$ and $(x, f(x))= (Z_{S},1)$, where $Z_{S}=\sum_{i\in\mathcal{S}} \exp{(- e_{i}/\kt)}$ is the system's partition function. Geometrically, to have a transition $\rho_{S} \to \rho_{S}'$ under a thermal operation, $\rho_{S}$ thermomajorization curve must lie above or on the curve of $\rho_{S}'$. (See Fig. \ref{fig:thermomajorizationcurvemainbody}.)

\begin{figure}[t]
    \centering
    \scalebox{0.8}{
    \begin{tikzpicture}
    \draw[black, very thick](0,0) rectangle(8,4) node [label={[shift={(0,-4.8)}]:{\small $Z_{S} $}}] {};
    \draw[blue, very thick](0,0)--(2,4);
    \draw[blue, very thick](2,4)--(8,4);
    \draw[red, very thick](0,0)--(8,4);
    \draw[teal, very thick](0,0)--(1,1);
    \draw[teal, very thick](1,1)--(3,2.332);
    \draw[teal, very thick](3,1+4*0.333)--(5,1+4*0.333+1);
    \draw[teal, very thick](5,1+4*0.333+1)--(8,4);
    
    \draw[dashed](2,4)--(2,0) node [label={[shift={(0,-0.8)}]:{ $e^{-\beta e_{1}}$}}] {};
    \draw[dashed](5,4)--(5,0) node [label={[shift={(0,-0.8)}]:{ $e^{-\beta e_{1}}+e^{-\beta e_{2}}$}}] {};
    \filldraw[blue] (0-0.07,0-0.07) rectangle (0+0.07,0+0.07);
    \filldraw[blue] (2-0.07,4-0.07) rectangle (2+0.07,4+0.07);  
    \filldraw[blue] (5-0.07,4-0.07) rectangle (5+0.07,4+0.07);
    \filldraw[blue] (7-0.07,4-0.07) rectangle (7+0.07,4+0.07);
    \filldraw[blue] (8-0.07,4-0.07) rectangle (8+0.07,4+0.07);
    \filldraw[red] (0,0) circle (2pt);
    \filldraw[red] (3,1.5) circle (2pt);
    \filldraw[red] (5,2.5) circle (2pt);
    \filldraw[red] (7,3.5) circle (2pt);
    \filldraw[red] (8,4) circle (2pt);
    \filldraw[teal] (0-0.07,0-0.040) -- (0+0.07,0-0.040) -- (0,0.080) -- cycle;
    \filldraw[teal] (1-0.07,1-0.040) -- (1+0.07,1-0.040) -- (1,1+0.080) -- cycle;
    \filldraw[teal] (3-0.07,2.332-0.040) -- (3+0.07,2.332-0.040) -- (3,2.332+0.080) -- cycle;
    \filldraw[teal] (5-0.07,3.332-0.040) -- (5+0.07,3.332-0.040) -- (5,3.332+0.080) -- cycle;
    \filldraw[teal] (8-0.07,4-0.040) -- (8+0.07,4-0.040) -- (8,4+0.080) -- cycle;
    \draw[black](2,4)--(0,4) node [label={[shift={(-0.45,-0.35)}]:{1}}] {};
    \filldraw[blue] (6.5-0.07,2-0.07) rectangle (6.5+0.07,2+0.07) node [label={[shift={(0.45,-0.4)}]:{ $\rho_{S}$}}] {};
    \filldraw[teal] (6.5-0.07,1.5-0.040) -- (6.5+0.07,1.5-0.040) -- (6.5,1.5+0.080) -- cycle node [label={[shift={(0.6,-0.3)}]:{ $\sigma_{S}$}}] {};
    \filldraw[red] (6.5,1) circle (2pt) node [label={[shift={(0.5,-0.3)}]:{ $\tau_{S}$}}] {};
    \end{tikzpicture}
    }
\caption{Thermomajorization curves of states: We show thermomajorization curves of three states $\rho_{S}$, $\sigma_{S}$, and $\tau_{S}$. $\rho_{S}$ is a pure state, $\sigma_{S}$ a general state and $\tau_{S}$ the Gibbs state. Applying the criterion, we can have transitions $\rho_{S}\to\sigma_{S}/\tau_{S}$ and $\sigma_{S}\to\tau_{S}$ under thermal operations.}
\label{fig:thermomajorizationcurvemainbody}
\end{figure}

Now, we are ready to study work extraction bounds in the single-shot regime. Consider a two-level work reservoir with Hamiltonian $H_{W}=W_{0} \ra{W_{0}}\la{W_{0}} + W_{1} \ra{W_{1}}\la{W_{1}}$. For a work extraction transition $(\rho_{S} \otimes |W_{0}\rangle \langle W_{0}|, H_{S}+H_{W}) \to (\tau_{S} \otimes |W_{1}\rangle \langle W_{1}|, H_{S}+H_{W})$ 
 to occur in single-shot thermodynamics, $\rho_{S} \otimes |W_{0}\rangle \langle W_{0}|$ must thermomajorize $\tau_{S} \otimes |W_{1}\rangle \langle W_{1}|$ and we have:
\begin{align}
    W & = W_{1}-W_{0}
     \leq D_{0}(\rho_{S} \| \tau_{S})
    ~.
\end{align}
(See Appendix \ref{appendix:thermaloperations} for details.)

In this case, the maximum work extractable from a state $(\rho_{S}, H_{S})$ cannot achieve the upper bound $-\Delta F_{S}$, because a two-level nanoscale work reservoir cannot approach zero entropy production for every work extraction. By contrast, stochastic thermodynamics approaches zero entropy production by employing a quasistatic process connecting the initial and final states \cite{deffner2019quantum}. 


Next, let us address how to compute the entropy production in the single-shot regime. The entropy production is still defined as in Eq. \eqref{eqn:entropyproduction}. Consider an energy preserving unitary operation such that $Q = -\Delta U_{S},$ where:
\begin{align}
\Delta U_{S} = \kt\Big(- \mathrm{Tr}(\rho_{S}' \log \tau_{S}) +\mathrm{Tr} (\rho_{S} \log \tau_{S})\Big)
  ~.
\end{align}
(Here, we assume there is no work reservoir. But if we wish to include one, we treat the work reservoir as part of the system.) Then we can write the entropy production of Eq. \eqref{eqn:entropyproduction} in an information-theoretic form \cite{riechers2021initial,reeb2014improved}: 
\begin{align}
\label{eqn:entropyproductioninfotheoritic}
    \Sigma & =-\Delta U_{S}/T +\Delta S_{S} \\
    & = D(\rho_{S} \| \tau_{S}) - D(\rho_{S}' \| \tau_{S})
    ~.
\end{align}
This represents the entropy produced when the system undergoes a Gibbs-preserving thermal operation, whose steady state $\tau_{S}$ produces zero entropy.  In essence, when there is no work reservoir to guide the transformation, any relaxation towards equilibrium corresponds to irreversibility.

Let thermal operation $\mathcal{E}$ transform $\rho_{S}$ to $\rho_{S}'$: $\mathcal{E}(\rho_{S})=\rho_{S}'$.  This thermal operation preserves the Gibbs state, such that $\mathcal{E}(\tau_{S}) = \tau_{S}$, and from the data processing inequality \cite{wilde2013quantum},
we have:
\begin{align}
    D(\rho_{S} \| \tau_{S}) & \geq D(\mathcal{E}(\rho_{S})\|\mathcal{E}(\tau_{S})) \\ 
    & = D(\rho_{S}'\|\tau_{S})
    ~.
\end{align}
Entropy production is always nonnegative in single-shot thermodynamics. Now, we are ready to state a theorem on approaching zero entropy production at the nanoscale.

\begin{theorem}
\label{thm:coincide}
Consider a $d-$dimensional system with Hamiltonian $H$. Given two states $\rho$ and $\sigma$, the following are equivalent:
\begin{enumerate}[label=(\alph*)]
\item \label{Coincide1} The thermomajorization curves of states $\rho$ and $\sigma$ coincide.
\item \label{Coincide2} There exists a thermal operation $\mathcal{E}$ such that $\mathcal{E}(\rho)$ can be arbitrarily close to $\sigma$ and the corresponding entropy production can be arbitrarily small.
\end{enumerate}
\end{theorem}


Theorem \ref{thm:coincide} is one of our main results. Appendix \ref{appendix:proof} gives the proof. Note that for two different states to have exactly same thermomajorization curve, there must be energy degeneracy in $H$ \cite{de2022geometric,mazurek2018decomposability}. Theorem \ref{thm:coincide} illustrates geometrically why the familiar thermodynamics bounds are not same as the bounds at the nanoscale. To approach the latter bounds, the entropy production needs to be arbitrarily small. Here, the work reservoir entropy change must be included:
\begin{align}
    \Sigma= \Delta S_{S} + \Delta S_{W} + \frac{Q}{T}
    ~.
\end{align}
Under deterministic work extraction, $\Delta S_{W}=0$ and the work reservoirs' initial and final states are pure states. They can only contract the system's thermomajorization curves along $x$-axis by a factor. And so, to approach zero entropy production, the system's initial thermomajorization curve must coincide with its final thermomajorization curve up to a contraction factor. This is not always possible. Fig. \ref{fig:example} depicts the situation. 

\begin{figure}[t]
    \centering
    \scalebox{1}{
    \begin{tikzpicture}
    \draw[black, very thick](0,0) rectangle(4,4) node [label={[shift={(0.4,-4.6)}]:{\small $2=Z_{S} $}}] {};
    \draw[blue, very thick](0,0)--(2,4);
    \draw[blue, very thick](2,4)--(4,4);
    \draw[red, very thick](0,0)--(4,4);
    \draw[green, very thick](0,0)--(2,4*0.666);
    \draw[green, very thick](2,4*0.666)--(4,4);
    
    \draw[dashed](2,4)--(2,0) node [label={[shift={(0.55,-0.6)}]:{\small $1=e^{-\beta 0}$ }}] {};
    \filldraw[blue] (0-0.07,0-0.07) rectangle (0+0.07,0+0.07);
    \filldraw[blue] (2-0.07,4-0.07) rectangle (2+0.07,4+0.07);  
    \filldraw[blue] (4-0.07,4-0.07) rectangle (4+0.07,4+0.07);
    \filldraw[red] (0,0) circle (2pt);
    \filldraw[red] (2,2) circle (2pt);
    \filldraw[red] (4,4) circle (2pt);
    \filldraw[green] (0-0.07,0-0.040) -- (0+0.07,0-0.040) -- (0,0.080) -- cycle;
    \filldraw[green] (4-0.07,4-0.040) -- (4+0.07,4-0.040) -- (4,4+0.080) -- cycle;
    \filldraw[green] (2-0.07,4*0.666-0.040) -- (2+0.07,4*0.666-0.040) -- (2,4*0.666+0.080) -- cycle;
    \draw[dashed](2,2)--(0,2) node [label={[shift={(-0.45,-0.45)}]:{\small $\frac{1}{2}$ }}] {};
    \draw[dashed](2,4*0.666)--(0,4*0.666) node [label={[shift={(-0.45,-0.45)}]:{\small $\frac{2}{3}$ }}] {};
    \draw[black](2,4)--(0,4) node [label={[shift={(-1.05,-0.45)}]:{\small $\sum_{i}p_{i}=1$ }}] {};
    \end{tikzpicture}
    }
\caption{Thermodynamics' bound cannot be achieved at the nanoscale: Consider a two-level system spanned by $\{\ra{0}, \ra{1}\}$ with $H_{S}=0$. The red circle is the thermomajorization curve of $\rho_{S}=\frac{1}{2}(\ra{0}\la{0}+\ra{1}\la{1})$. The blue square is the curve for $\rho_{S}'=\ra{0}\la{0}$. The green triangle is that of $\sigma_{S}=(\frac{1}{3}\ra{0}\la{0}+\frac{2}{3}\ra{1}\la{1})$. The red circle and blue square curves coincide with a two-level work reservoir. The corresponding transition $(\rho_{S},H_{S}) \to (\rho'_{S},H_{S})$ is the well-known Landauer's erasure. We can approach this bound arbitrarily closely. However, this cannot be done for the green triangle and blue square curves.}
\label{fig:example}
\end{figure}


\section{Beyond deterministic work}
\label{section:beyonddeterministicwork}

This section generalizes two-level work reservoirs in such a way that initial and final thermomajorization curves coincide. This achieves arbitrarily small entropy production for a transition. Before the general case, though, we first review an elementary example to give a simple picture.

\subsection{Example}

Consider Landauer's erasure with the initial distribution $\boldsymbol{p}_{S}=(\frac{1}{3},\frac{2}{3})$ stored in a two level system with trivial Hamiltonian $H=0$ and a four-level work reservoir with energy levels $\{W_{0},W_{1},W_{2},W_{3}\}$. We set the work reservoir's initial distribution to $\boldsymbol{p}_{W}=(r_{1}, r_{2}, 0, 0)$ and the final to $\boldsymbol{p}_{W}'=(0, 0, r_{1}, r_{2})$. Initially, the nonzero populations of the work reservoir are with the first half of energy levels and the final nonzero populations of the work reservoir are with the second set of energy levels. The work reservoir's entropy does not change overall. The total initial state is:
\begin{align}
\rho_{SW}=&(\frac{1}{3}\ra{0}\la{0}+\frac{2}{3}\ra{1}\la{1}) \otimes \\
& (r_{1}\ra{W_{0}}\la{W_{0}}+r_{2}\ra{W_{1}}\la{W_{1}})
\end{align}
and the final is:
\begin{align}
\rho_{SW}'=\ra{0}\la{0}\otimes (r_{1}\ra{W_{2}}\la{W_{2}}+r_{2}\ra{W_{3}}\la{W_{3}})
~.
\end{align}

First, consider the final state's thermomajorization curve. At most, it has two distinct slopes. For the two curves to coincide, the initial curve can contain at most two distinct slopes. One possibility is that the initial work reservoir's thermomajorization curve has one distinct slope. This leads to:
\begin{align}
    \frac{1}{3}r_{1}e^{\beta W_{0}} = \frac{1}{3}r_{2}e^{\beta W_{1}} &= r_{1}e^{\beta W_{2}} \\
    \frac{2}{3}r_{1}e^{\beta W_{0}} =\frac{2}{3}r_{2}e^{\beta W_{1}}&=r_{2}e^{\beta W_{3}} \\
    \frac{1}{3}r_{1}+\frac{1}{3}r_{2} &= r_{1} \\
    \frac{2}{3}r_{1}+\frac{2}{3}r_{2} &= r_{2}~.
\end{align}
The first two equations come from requiring the initial curve to have only two distinct slopes and the same slopes as the final curve's. And, the last two equations come from requiring the same $y$-coordinate change. Solving those equations gives:
\begin{align}
r_{1} & = \frac{1}{3} ~\text{and}~r_{2} = \frac{2}{3} \\
e^{-\beta W_{0}} & = a ,~ e^{-\beta W_{1}} = 2a ,~ e^{-\beta W_{2}} =3 a ,~\text{and}~
e^{-\beta W_{3}}=3a
    ,
\end{align}
where $a$ is an arbitrary positive number.

\begin{table*}[!t]
    \centering
    \scalebox{0.3}{
    \includegraphics[]{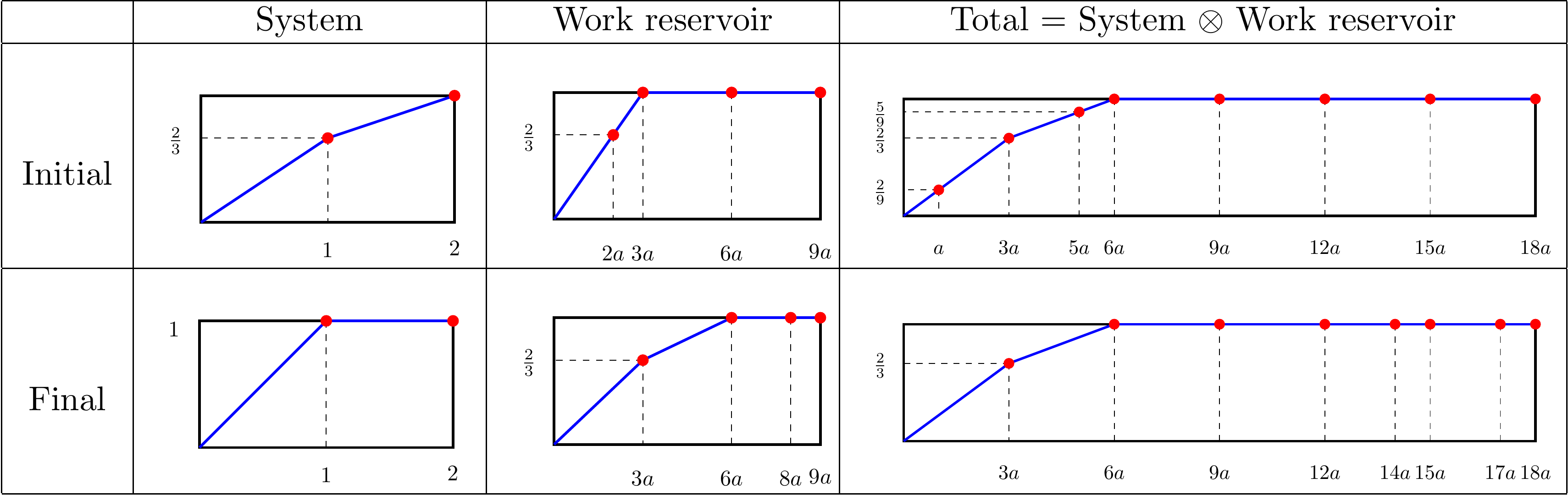}
    }
\caption{Efficient work reservoir for Landauer erasure: The first row shows thermomajorization curves of the initial system state $\rho_{S}=\frac{1}{3}\ra{0}\la{0}+\frac{2}{3}\ra{1}\la{1}$, the initial work state $\rho_{W}=\frac{1}{3}\ra{W_{0}}\la{W_{0}}+\frac{2}{3}\ra{W_{1}}\la{W_{1}}$, and the initial total state $\rho_{SW}=\rho_{S}\otimes\rho_{W}$. The second row shows thermomajorization curves of the final system state $\rho_{S}'=\ra{0}\la{0}$, the final work state $\rho_{W}'=\frac{1}{3}\ra{W_{2}}\la{W_{2}}+\frac{2}{3}\ra{W_{3}}\la{W_{3}}$, and the final total state $\rho'_{SW}=\rho_{S}'\otimes\rho_{W}'$.}
    \label{tab:lorenzcurveex}
\end{table*}

Table \ref{tab:lorenzcurveex} demonstrates that the initial and final curves coincide. The expected energy change in the work reservoir is:
\begin{align}
\avg{W} & =r_{1}(W_{2}-W_{0})+r_{2}(W_{3}-W_{1}) \\
  & =\kt \left( \frac{1}{3}\log\frac{1}{3}+\frac{2}{3}\log\frac{2}{3} \right)
  ~.
\end{align}
This is the system entropy change as expected. This demonstrates that energy levels $\boldsymbol{E}_{W}=\{W_{0},W_{1},W_{2},W_{3} \}$ with probability distributions $\boldsymbol{p}_{W}=(\frac{1}{3}, \frac{2}{3},0,0)$ and $\boldsymbol{p}'_{W}=(0,0,\frac{1}{3}, \frac{2}{3})$ form an efficient work reservoir for Landauer erasure with the initial distribution $\boldsymbol{p}_{S}=(\frac{1}{3},\frac{2}{3})$.

One subtlety to highlight is that, although the total curves coincide, with thermal operations we can only make the final state arbitrarily close to the desired state $\rho_{SW}'$. So, we cannot use exactly $\avg{W}$ to erase $\boldsymbol{p}_{S}$. Instead, we can use the amount of work arbitrarily close to $\avg{W}$ to erase $\boldsymbol{p}_{S}$ and then the corresponding entropy production will be arbitrarily small.

For simplicity, from now on we treat thermomajorization curves coinciding as the same as zero entropy production. Corresponding bounds on work can be computed by setting entropy production to be zero. However, we should keep in mind that the precise statement is that the entropy production can be arbitrarily small and the corresponding work can be arbitrarily close to the bounds.

A key observation from this example is that for the two total thermomajorization curves to coincide, the nonzero slope part of final work reservoir's curve must coincide with the nonzero slope part of the initial system's curve up to a scale constant. In the above example, the scale constant is $3a$. We further require that the nonzero slope part of initial work reservoir's curve coincide with the nonzero slope part of the final system's curve up to the same scale constant. Thus, the key step in constructing an efficient work reservoir for a state transformation is to find a suitable probability distribution for the work reservoir. And, then we can fine tune energy levels such that the work reservoir's thermomajorization curve can coincide with both the initial and final state's curves.

In the first example, the probability distribution we chose was $(\frac{1}{3},\frac{2}{3})$. We set the initial energy levels to be $W_{0}=-k_{B}T \log a$ and $W_{1}=-k_{B}T \log 2a$ and the final's to be $W_{2}=-k_{B}T \log 3a$ and $W_{3}=-k_{B}T \log 3a$. Under those parameters, the erasure is efficient; i.e., the entropy production vanishes.

\subsection{General efficient work reservoirs}

Now, we turn to develop efficient work reservoirs for arbitrary state transitions. First, we introduce a notation using tuples to aid in describing thermomajorization curves. Then, we present the definition of efficient work reservoirs and briefly discuss how to construct them. Recall that the thermomajorization curve $f_{\boldsymbol{p},\boldsymbol{E}}$ of a distribution $\boldsymbol{p}=\{p_i\}_{i \in \mathcal{S}}$ over the energy levels $\boldsymbol{E}=\{\epsilon_i\}_{i\in \mathcal{S}}$ can be derived from the collection of segments $\{(e^{-\beta \epsilon_i},p_i)\}_{i\in \mathcal{S}}$. Thermomajorization curve orders the segments from highest slope---the slope of $i$-th element is $p_i e^{\beta \epsilon_i}$---to lowest and then concatenates them end to end.

Consider a coarse-graining function $\lambda:\mathcal{S} \rightarrow \mathcal{S}'$ that defines a new distribution and energy landscape: $\boldsymbol{p}'=\lambda(\boldsymbol{p})=\{p'_j \}_{j \in \mathcal{S}'}$ and energy landscape $\boldsymbol{E}'=\lambda(\boldsymbol{E})=\{\epsilon'_j \}_{j \in \mathcal{S}'}$  via:
\begin{align}
p'_j & = \sum_{i \in \lambda^{-1}(j)} p_i \\
   e^{-\beta \epsilon'_j} & = \sum_{i \in \lambda^{-1}(j)} e^{-\beta \epsilon_i}
  ~,
\end{align}
where:
\begin{align}
    \lambda^{-1}(j) \equiv \{i | i \in \mathcal{S} , \lambda(i)=j\}
    ~.
\end{align}

If $\lambda$ only coarse-grains elements of $(\boldsymbol{p},\boldsymbol{E})$ whose segments have the same slope---meaning $\lambda(i)=\lambda(i')$ implies $p_ie^{\beta \epsilon_i}=p_{i'}e^{\beta \epsilon_{i'}}$---then the coarse-grained distribution and energies $(\lambda(\boldsymbol{p}),\lambda(\boldsymbol{E}))$ have the same thermomajorization curve $f_{\lambda(\boldsymbol{p}),\lambda(\boldsymbol{E})}=f_{\boldsymbol{p},\boldsymbol{E}}$. The segments $(e^{-\beta \epsilon_i},p_i)$ and $(e^{-\beta \epsilon_{i'}},p_{i'})$ of elements $i$ and $i'$ with the same slope in the thermomajorization curve comprise a long line segment with (width, height) = $(e^{-\beta \epsilon_{i}}+e^{-\beta \epsilon_{i'}},p_i+p_{i'})$.

Suppose $\lambda$ coarse-grains all segments with the same slopes. After the coarse-graining, the thermomajorization curve has $n$ distinct slopes, excluding the segments with slope zero. Let $\# f_{\boldsymbol{p},\boldsymbol{E}}=n$ denote the number of distinct slopes in $f_{\boldsymbol{p},\boldsymbol{E}}$ and $n$ tuples $\boldsymbol{f}_{\boldsymbol{p},\boldsymbol{E}}=\{(y_{i},k_{i})\}_{i=1}^{n}$ represent $f$ where $k_{i}$ is the $i$-th distinct slope and $y_{i}$ is the corresponding $y-$coordinate change. In some cases, we allow repeating slopes in $\boldsymbol{f}_{\boldsymbol{p},\boldsymbol{E}}$. 

For a composite system, the joint thermomajorization curve is constructed as follows. Given one distribution $\boldsymbol{p}_{S}=\{p_{i}\}_{i\in \mathcal{S}}$ over energy levels $\boldsymbol{E}_{S} =\{e_{i}\}_{i\in \mathcal{S}}$ with thermomajorization curve $\boldsymbol{f}_{\boldsymbol{p}_{S},\boldsymbol{E}_{S}}=\{(x_{i},k_{i})\}_{i\in\mathcal{S}}$ and another distribution $\boldsymbol{p}_{S'}=\{q_{i}\}_{i\in \mathcal{S}'}$ over energy levels $\boldsymbol{E}_{S'} =\{h_{i}\}_{i\in \mathcal{S}'}$ with thermomajorization curve $\boldsymbol{f}_{\boldsymbol{q}_{S'},\boldsymbol{H}_{S'}}=\{(y_{i},m_{i})\}_{i\in\mathcal{S'}}$, then the composite configuration is the probability distribution $\boldsymbol{p}_{SS'}$ over energy levels $\boldsymbol{E}_{SS'}$ where:
\begin{align}
    \boldsymbol{p}_{SS'} &= \{p_{i}q_{j}\}_{i\in \mathcal{S}, j \in \mathcal{S}'} \\ 
    \boldsymbol{E}_{SS'} &= \{e_{i}+h_{j}\}_{i\in \mathcal{S}, j \in \mathcal{S}'} ~,
\end{align}
and:
\begin{align}\label{eqn:compositethermalcurve}
    \boldsymbol{f}_{\boldsymbol{p}_{SS'},\boldsymbol{E}_{SS'}}=\{(x_{i}y_{j},k_{i}m_{j})\}_{i\in\mathcal{S},j\in\mathcal{S'}}~.  
\end{align}
Slopes may repeat in $\boldsymbol{f}_{\boldsymbol{p}_{SS'},\boldsymbol{E}_{SS'}}$.


With this enhanced notation, we now define multi-level work reservoirs.

\begin{defi}\label{def:mworkreservoir}(Multi-level Work Reservoirs)
A $2d-$level work reservoir $(\boldsymbol{p}_W, \boldsymbol{p}_W',\boldsymbol{E}_W)$ for a state transition $\boldsymbol{p}_{S} \to \boldsymbol{p}_{S}'$ in a system with energy levels $\boldsymbol{E}_S = \{e_s\}_{s \in \mathcal{S}}$ has initial distribution $\boldsymbol{p}_W = \{q_w \}_{w \in \mathcal{W}}$, final distribution $\boldsymbol{p}_W' = \{q_w' \}_{w \in \mathcal{W}}$, and energy eigenstates $\boldsymbol{E}_W=\{\epsilon_w \}_{w \in \mathcal{W}}$. Here, $\boldsymbol{p}_W$ and $\boldsymbol{p}_W'$ have the form of $\boldsymbol{p}_W=(\boldsymbol{r}, \boldsymbol{0})$ and $\boldsymbol{p}_W'=(\boldsymbol{0}, \boldsymbol{r})$, where $\boldsymbol{r}$ is a $d-$dimension probability distribution. The initial configuration of the system with the reservoir is $(\boldsymbol{p}_{SW},\boldsymbol{E}_{SW})$, where:
\begin{align}
    \boldsymbol{p}_{SW}&=\{ p_s q_w\}_{s \in \mathcal{S},w\in \mathcal{W}} \\
    \boldsymbol{E}_{SW}&=\{ e_s + \epsilon_w\}_{s \in \mathcal{S},w\in \mathcal{W}}
    ~.
\end{align}
The final configuration is $(\boldsymbol{p}'_{SW},\boldsymbol{E}_{SW})$, where:
\begin{align}
\boldsymbol{p}'_{SW}&=\{ p'_s q'_w \}_{s \in \mathcal{S},w\in \mathcal{W}} \\
    \boldsymbol{E}_{SW}&=\{ e_s+\epsilon_w \}_{s \in \mathcal{S},w\in \mathcal{W}}
    ~.
\end{align}
\end{defi}

This requires $\boldsymbol{p}_W$ and $\boldsymbol{p}_W'$ to have the forms $\boldsymbol{p}_W=(\boldsymbol{r}, \boldsymbol{0})$ and $\boldsymbol{p}_W'=(\boldsymbol{0}, \boldsymbol{r})$ so that the overall work reservoir's entropy change vanishes. This satisfies the stochastic thermodynamics' entropyless assumption for work reservoirs \cite{semaan2022homeostatic}. Furthermore, from our definition the initial nonzero distribution in the work reservoir occupies the first half of energy levels and the final occupies the second half of the energy levels. This leads immediately to the following definition.

\begin{defi}(Efficient Work Reservoirs)
\label{defi:eworkreservoir}
A work reservoir $(\boldsymbol{p}_W, \boldsymbol{p}_W',\boldsymbol{E}_W)$ is efficient for a state transition $\boldsymbol{p}_{S} \to \boldsymbol{p}_{S}'$ in a system with energy levels $\boldsymbol{E}_S $ if the thermomajorization curves of $(\boldsymbol{p}_{SW},\boldsymbol{E}_{SW}) $and $(\boldsymbol{p}'_{SW},\boldsymbol{E}_{SW})$ coincide.
\end{defi}

The previous example showed that the key to constructing an efficient work reservoir is to find a probability distribution $\boldsymbol{r}$ and energy levels $\boldsymbol{E}_{W}$ for work reservoirs such that its curve coincides with the final state's curve with the distribution $(\boldsymbol{r},\boldsymbol{0})$ and with the initial state's curve with the distribution $(\boldsymbol{0}, \boldsymbol{r})$ up to the same constant. Then the efficient work reservoir's initial and final thermomajorization curves mimic the system's final and initial thermomajorization curves as shown in Fig. \ref{fig:efficientreservoir}. 

\begin{figure}
    \centering
    \scalebox{1.0}{\includegraphics[width=\columnwidth]{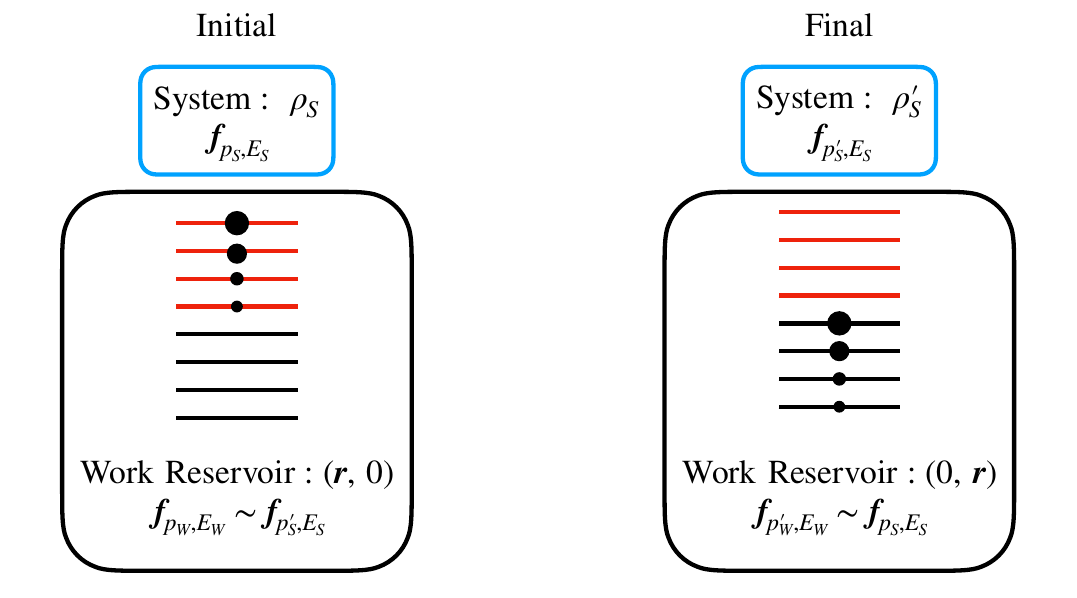}
    }
    \caption{One way for two total thermomajorization curves to coincide is to find a probability distribution $\boldsymbol{r}$ and energy levels $\boldsymbol{E}_{W}$ such that when $\boldsymbol{r}$ occupies the first half energy levels---i.e., $\boldsymbol{p}_{W}=(\boldsymbol{r},0)$---the thermomajorization curve $\boldsymbol{f}_{\boldsymbol{p}_{W},\boldsymbol{E}_{W}}$ coincides with the final system state's curve $\boldsymbol{f}_{\boldsymbol{p}_{S}',\boldsymbol{E}_{S}}$ up to a contraction factor---denoted  $\boldsymbol{f}_{\boldsymbol{p}_{W},\boldsymbol{E}_{W}}\sim\boldsymbol{f}_{\boldsymbol{p}_{S}',\boldsymbol{E}_{S}}$. When the probability distribution occupies the second half energy levels---i.e., $\boldsymbol{p}_{W}'=(0,\boldsymbol{r})$---the thermomajorization curve $\boldsymbol{f}_{\boldsymbol{p}_{W}',\boldsymbol{E}_{W}}$ coincides with the initial system state's curve $\boldsymbol{f}_{\boldsymbol{p}_{S},\boldsymbol{E}_{S}}$ up to the same contraction factor. With this, the two total curves coincide: $\boldsymbol{f}_{\boldsymbol{p}_{SW},\boldsymbol{E}_{SW}}=\boldsymbol{f}_{\boldsymbol{p}_{SW}',\boldsymbol{E}_{SW}}$.
    } 
    \label{fig:efficientreservoir}
\end{figure}

\subsection{Work extraction and state formation reservoirs}

We first study how to construct efficient work reservoirs for two kinds of state transitions: work extractions and state formations. For work extraction $(\boldsymbol{p}_{S},\boldsymbol{E}_{S})\to (\boldsymbol{\tau}_{S},\boldsymbol{E}_{S})$, suppose there are $m$ distinct slopes in the  thermomajorization curve $f_{\boldsymbol{p}_{S},\boldsymbol{E}_{S}}$ and $\boldsymbol{f}_{\boldsymbol{p}_{S},\boldsymbol{E}_{S}}=\{(r_{i},a_{i})\}_{i=1}^{m}$, where $m$ is the number of distinct slopes in the thermomajorization curve of the system. We now show that a work reservoir must have a dimension greater than $2(m-1)$ to achieve efficient work extraction. 

To see this, assume that an efficient work reservoir has dimension $2d \leq 2(m-1)$. Now, let the initial work reservoir probability distribution be $\boldsymbol{p}_{W}=(\boldsymbol{r},\boldsymbol{0})$,  the corresponding thermomajorization curve have $a$ distinct slopes, the final work reservoir probability distribution be $\boldsymbol{p}_{W}'=(\boldsymbol{0},\boldsymbol{r})$, and the corresponding thermomajorization curve have $b$ distinct slopes. Since the dimension of $\boldsymbol{r}$ is $d$. Then, we have $ a,b\leq d \leq (m-1)$. 

The final total probability distribution is $\boldsymbol{p}'_{SW}=\boldsymbol{\tau}_{S} \otimes \boldsymbol{p}'_{W}$. We have $\# f_{\boldsymbol{p}'_{S},\boldsymbol{E}_{S}} = 1$ and $\# f_{\boldsymbol{p}'_{W},\boldsymbol{E}_{W}} = b$. The number of distinct slopes of the thermomajorization curve $f_{\boldsymbol{p}'_{SW}, \boldsymbol{E}_{SW}}$ is $b$. The initial total probability distribution is $\boldsymbol{p}_{SW}=\boldsymbol{p}_{S} \otimes \boldsymbol{p}_{W}$. Since the number of distinct slopes in $\boldsymbol{p}_{S}$'s thermomajorization is $m$, we have $\# f_{\boldsymbol{p}_{SW},\boldsymbol{E}_{SW}}\geq m $. The equality holds if and only if the number of the segments of $\boldsymbol{p}_{W}$'s thermomajorization is 1. Since we have $b\leq m-1 < m$, it is impossible for curve $f_{\boldsymbol{p}_{SW}, \boldsymbol{E}_{SW}}$ to coincide with curve $f_{\boldsymbol{p}'_{SW}, \boldsymbol{E}_{SW}}$. Hence, the dimension of the efficient work reservoir is at least $2m$. 

With a $2m$ dimension work reservoir, we choose the probability distribution to be $\boldsymbol{r}=(r_{1}, \cdots, r_{m})$. We fine tune the first-half energy levels such that the initial work reservoir's curve only contains one slope which coincides with the final thermal state of the system's curve up to a constant. For the second-half energy levels, they are fine tuned such that the final work reservoir's curve coincides with $f_{\boldsymbol{p}_{S},\boldsymbol{E}_{S}}$ up to the same constant. (See TABLE \ref{table:minimumworkreservoir}).

The detailed calculation follows. Suppose the energy levels of the work reservoir are $\boldsymbol{E}_{W}=\{\epsilon_{1}, \cdots, \epsilon_{m}, \epsilon_{1}', \cdots, \epsilon_{m}'\}$. For the energy levels $\{\epsilon_{1}, \cdots, \epsilon_{m}\}$, we require:
\begin{align}
    e^{\beta \epsilon_{i}}=\frac{c}{r_{i}}
    ~,
\end{align}
where $c$ can be an arbitrary positive number. For the energy levels $\{\epsilon_{1}', \cdots, \epsilon_{m}'\}$, we stipluate:
\begin{align}
    e^{\beta \epsilon_{i}'}={c}{Z_{S}}\frac{a_{i}}{{r}_{i}}
    ~.
\end{align}
With our notation, we can verify that the two final curves coincide. For the initial setup:
\begin{align}
\boldsymbol{f}_{\boldsymbol{p}_{S},\boldsymbol{E}_{S}}&=\{(r_{i},a_{i})\}_{i=1}^{m}~,\\
\boldsymbol{f}_{\boldsymbol{p}_{W},\boldsymbol{E}_{W}}&=\{(1,c)\}~.
\end{align}
And for final setup:
\begin{align}
\boldsymbol{f}_{\boldsymbol{p}'_{S},\boldsymbol{E}_{S}}&=\{(1,1/Z_{S})\}~,\\
\boldsymbol{f}_{\boldsymbol{p}'_{W},\boldsymbol{E}_{W}}&=\{(r_{i},cZ_{S}a_{i})\}_{i=1}^{m}~.
\end{align}
From Eq. \eqref{eqn:compositethermalcurve} we have:
\begin{align}
\boldsymbol{f}_{\boldsymbol{p}_{SW},\boldsymbol{E}_{SW}}
 & =\boldsymbol{f}_{\boldsymbol{p}'_{SW},\boldsymbol{E}_{SW}} \\
 & =  \{(r_{i},ca_{i})\}_{i=1}^{m}~,
\end{align}
which means the two final total curves indeed coincide.
The energy change in this work reservoir is:
\begin{align}
W & = \sum_{i=1}^{m}r_{i}(\epsilon'_{i}-\epsilon_{i}) \\
   & = \kt D_{1}(\boldsymbol{p}_{S}||\boldsymbol{\tau}_{S})
  ~.
\end{align}
It is not hard to prove that this is the unique $2m$-dimensional efficient work reservoir for $\rho_{S}$ work extraction.

Since entropy production vanishes, we can use the same work reservoir to form the state $\boldsymbol{\tau}_{S} \to \boldsymbol{p}_{S}$. Hence, the minimal dimension of the efficient work reservoir for both work extraction and state formation is equal to $2\cdot\# f_ {\boldsymbol{p}_{S},\boldsymbol{E}_{S}}$. Appendix \ref{appendix:monoid} goes on to construct thermomajorization curves of all possible efficient work reservoirs for state formation and work extraction from the minimal efficient work reservoirs.

\begin{table}[!t]
    \scalebox{0.276}{
  \includegraphics[]{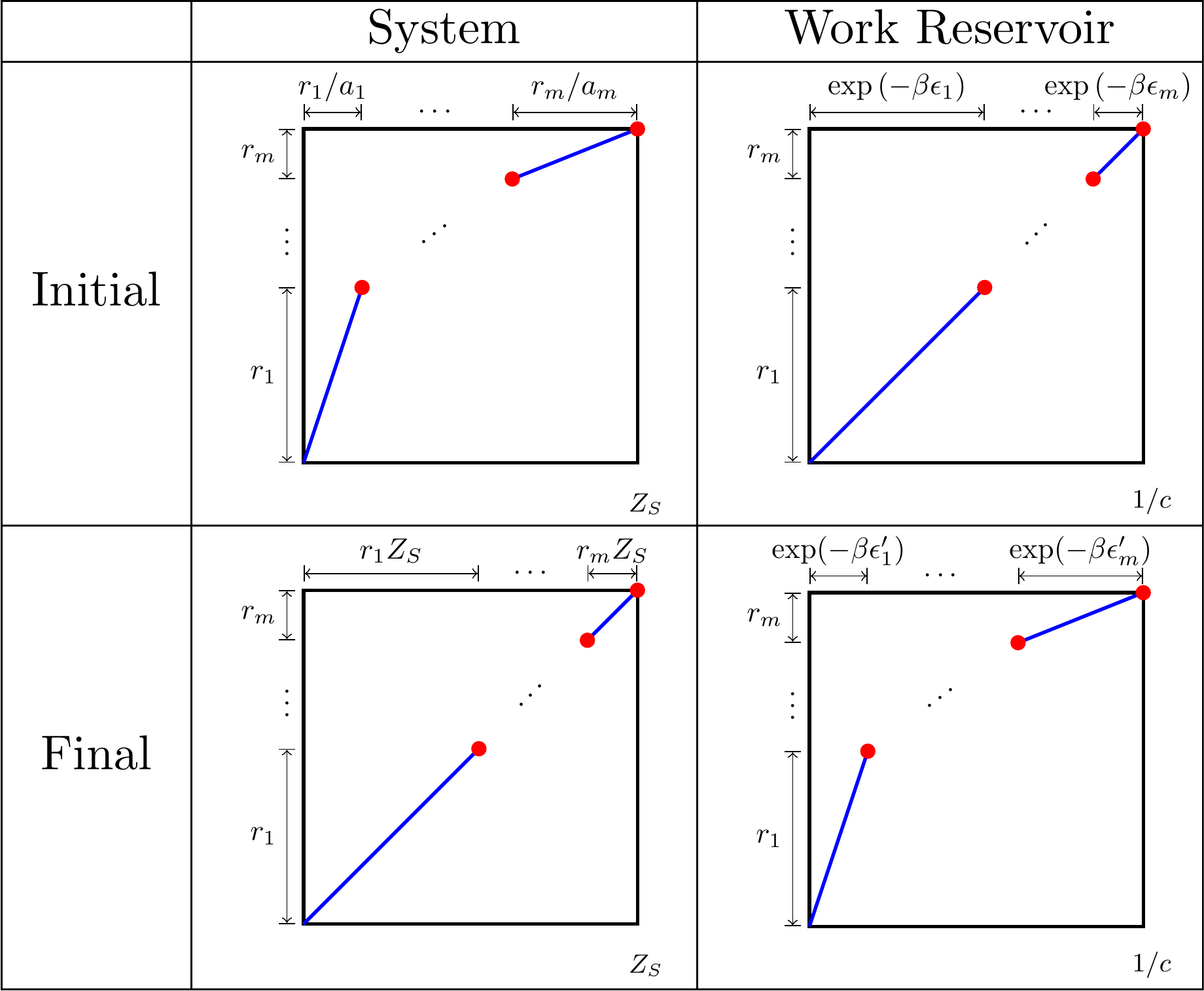}
  }
\caption{Initial and final thermomajorization curves for the efficient work extraction reservoir, ignoring the zero-slope parts.}
\label{table:minimumworkreservoir}
\end{table}

\subsection{Efficient reservoirs exist}

We will not develop all possible efficient work reservoirs for general state transitions here, though. Nonetheless, the next theorem establishes the existence of efficient work reservoirs for them---our second main result. 
\begin{theorem}\label{thm:existenceeffworkreservoir}
    For two general $n-$dimension states $\boldsymbol{p}_{S}$ and $\boldsymbol{p}'_{S}$ over energy levels $\boldsymbol{E}_{S}$, there exists a work reservoir $(\boldsymbol{p}_{W},\boldsymbol{p}'_{W},\boldsymbol{E}_{W})$ such that the thermomajorization curves of $(\boldsymbol{p}_{SW},\boldsymbol{E}_{SW}) $and $(\boldsymbol{p}'_{SW},\boldsymbol{E}_{SW})$ coincide.
\end{theorem}
Appendix \ref{appendix:detailconstruction} gives the details on how to construct the probability distribution and energy levels for efficient work reservoirs.

Here, we discuss several properties and applications of efficient work reservoirs. If $\boldsymbol{E}_{W}=\{\epsilon_{1},\cdots, \epsilon_{N}, \epsilon_{1}',\cdots, \epsilon_{N}'\}$ determines the energy levels for an efficient work reservoir with probability transition $(\boldsymbol{r},\boldsymbol{0})\to (\boldsymbol{0}, \boldsymbol{r})$, then $\boldsymbol{E}'_{W}=\{\epsilon_{1}+c, \cdots, \epsilon_{N}+c, \epsilon_{1}'+c, \cdots, \epsilon_{N}'+c\}$ gives the energy levels of an efficient work reservoir with the same probability distribution, where $c$ is a constant. This shows that efficient work reservoirs have translational symmetry. That is, only gaps between energy levels in efficient work reservoirs matter. 

Since our efficient work reservoirs have more than two levels, the work fluctuates. The entropy production with efficient work reservoirs could be arbitrarily small. The variance of the work, however, could be greater than nonefficient work reservoirs. This can be seen by noting that the work variance is 0 in two-level work reservoirs since the work is deterministic, while the work variance in efficient work reservoirs is greater than 0.

Consider an example. Suppose the system is three dimensional with trivial Hamiltonian $H=0$ and the initial distribution is $(\frac{1}{3}, \frac{2}{3}, 0)$. Using a two-level work reservoir to harness work from this system, the extractable work is $W_{\text{2-level}}=\kt \log{3/2}$ and the work variance is 0. The zero variance is due to the fact that the work is deterministic in a two level work reservoir. If we use an efficient work reservoir to harness work from this system, though, the average work is $W_{\text{efficient}}=\kt(\log 3-H(1/3))$, where $H(\cdot)$ is the binary entropy function. The work variance, however, is nonzero. We have $W_{\text{efficient}}>W_{\text{2-level}}$. This example shows us that for a protocol with nonzero entropy production, the work variance might be less compared to a protocol with zero entropy production.

For transitions under time-dependent Hamiltonians, we introduce a clock system \cite{horodecki2013fundamental}. Suppose the initial and final Hamiltonians are $H_{S}$ and $H_{S}'$, respectively. The total Hamiltonian including the clock system is:
\begin{align}
    H = H_{S} \otimes \ra{0}\la{0} + H_{S}' \otimes \ra{1}\la{1}
    ~.
\end{align}
With the clock system, we require that any transition to be $\rho_{S} \otimes \ra{0}\la{0} \to \rho_{S}' \otimes \ra{1}\la{1}$. In this, the Hamiltonian changes from $H_{S}$ to $H_{S}'$. Appendix \ref{appendix:examplesofreservoir} presents two examples of efficient work reservoirs for nontrivial Hamiltonians and for time-dependent Hamiltonian state transitions. 

One of applications of efficient work reservoirs is to build a quantum engine that approaches Carnot efficiency. Suppose we pick a two-dimension system spanned by $\{\ra{0},\ra{1}\}$ with Hamiltonian $H_{\mathrm{eng}}=\epsilon \ra{1}\la{1}$. The engine functions with a hot bath at temperature $T_{H}$ and a cold bath at temperature $T_{C}$. Initially, the work reservoir is in its Gibbs state at temperature $T_{C}$. First, it is brought to the hot bath, interacts with the hot bath to extract work and ends up in Gibbs state at temperature $T_{H}$. Then, it is brought to the cold bath, interacts with it to extract work and ends up in Gibbs state at temperature $T_{C}$ finishing the cycle. If we use efficient work reservoirs to extract work, then the entropy production is arbitrarily close to 0. Then the engine's efficiency approaches Carnot efficiency $1- T_{C} / T_{H}$. Appendix \ref{appendix:engine} gives the details for constructing the work reservoir's probability distribution and energy levels.

\begin{figure}
    \centering
    \includegraphics[width=\columnwidth]{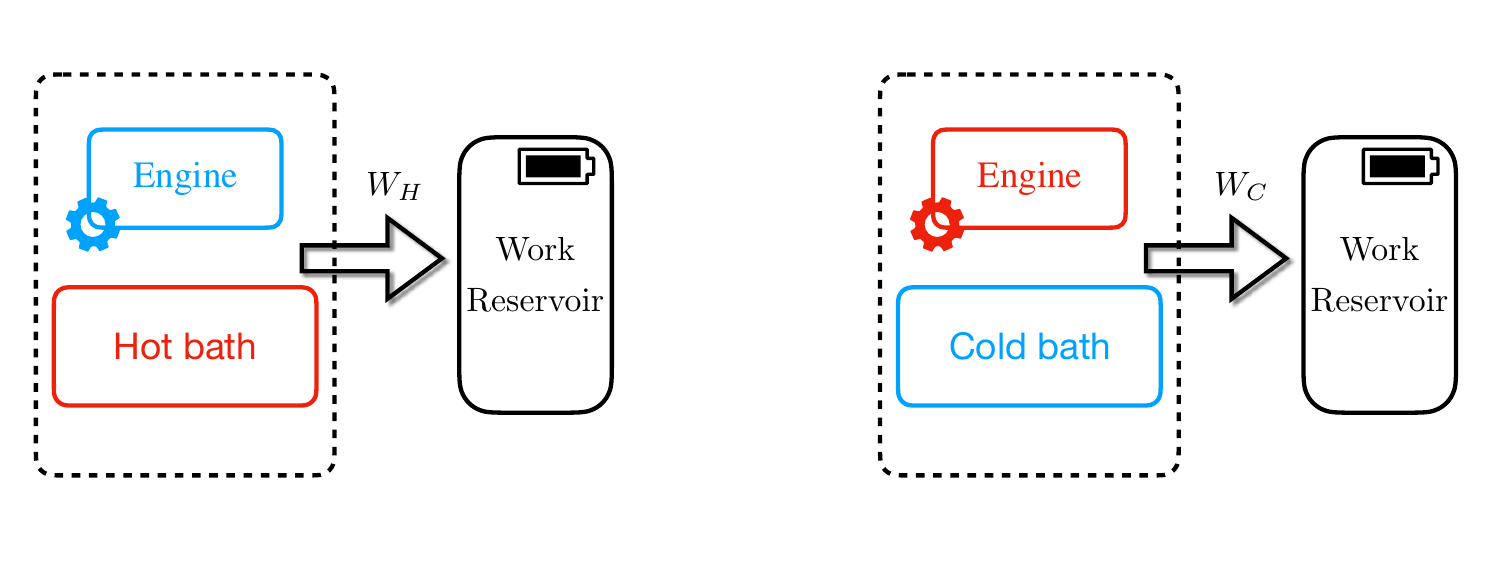}
    \caption{Work are stored into the reservoir during the engine and work reservoir interacting with the hot bath and the cold bath. First, the engine and reservoir interacts with the hot bath. The engine begins with the cold Gibbs state and ends with the hot Gibbs state. The amount of work $W_H = k_{B}T_{H} D_{1}(\tau_{C}||\tau_{H})$ is stored in the work reservoir. Then the engine and reservoir are brought to the cold bath. Similarly, the amount of work $W_C = k_{B}T_{C} D_{1}(\tau_{H}||\tau_{C})$ is stored in the work reservoir during the interaction with the cold bath.}
    \label{fig:engine}
\end{figure}

\section{Catalyzed work reservoirs}
\label{section:efficientworkreservoirwithcatalysts}

The development to this point was limited to noncatalytic scenarios. The following explores efficient work extraction with the aid of catalysts. Here, the intention is not to surpass the bound set by free energy differences. Rather, we ask whether we can extract work without dissipation by using a smaller work reservoir with catalysts. 

The main result in catalytic thermal operations is that the transition from state $\rho_{S}$ to $\rho_{S}'$ is possible through a \emph{catalytic thermal operation}---denoted $\rho_{S} \overset{CTO}{\longrightarrow} \rho_{S}'$---if and only if $D_{\alpha}(\rho_{S} \| \tau_{S}) \geq D_{\alpha}(\rho_{S}' \| \tau_{S})$, for all $\alpha \in \mathbb{R}$ \cite{brandao2015second}. The next theorem shows that catalysts do not help reach zero entropy production.

\begin{theorem}\label{thm: catalystwith0entropyproduction}
Consider a system with Hamiltonian $H$ and a catalyst state $c$ with Hamiltonian $H_{c}$. If state $\rho$ can be converted into a state that is arbitrarily close to state $\sigma$ through a thermal operation with the catalyst $c$ under arbitrarily small entropy production, then the transition can be achieved through a noncatalytic thermal operation. 
\end{theorem}
Appendix \ref{appendix:proof} gives the proof.
This shows that close to the zero dissipation regime, thermal operations and catalytic thermal operations are equivalent. 
Theorem \ref{thm: catalystwith0entropyproduction} provides yet another criterion for checking if two thermomajorization curves coincide.

\begin{theorem}\label{thm:coincide2}
Given a system with Hamiltonian $H$ and states $\rho$ and $\sigma$, the following are equivalent: 
\begin{enumerate}[label=(\alph*)]
    \item \label{Coincide21} Thermomajorization curves of $\rho$ and $\sigma$ coincide.
    \item \label{Coincide22} $D_{\alpha}(\rho \| \tau) = D_{\alpha}(\sigma \| \tau),\  \text{~for~all~} \alpha \in \mathbb{R}$.
\end{enumerate}
\end{theorem}
Again, we place the proof in Appendix \ref{appendix:proof}. It seems the catalysts are useless if we require the entropy production to be arbitrarily small. However, we find that if catalysts are allowed to correlate states in a trivial Hamiltonian, every state transition's entropy production can be reduced to 0; see Appendix \ref{Appendix:correlatedcatalyst}.

\section{Discussion}

In stochastic thermodynamics, it is well-known that the maximal extractable work from a state transition $\rho_{S} \to \rho_{S}'$ is the (negative) nonequilibrium free energy difference. The maximum is approached when the dissipation is arbitrarily small. However, as we showed, zero dissipation with two-level work reservoirs cannot always be approached in single-shot thermodynamics. This is due to the fact that, with two-level work reservoirs, we can only contract a thermomajorization curve by a factor. Two-level work reservoirs are not powerful enough to approach zero dissipation for every state transition.

To remove this restriction, we generalized two-level work reservoirs to multi-level work reservoirs. The extractable work is then defined as the difference in the expectation values of work reservoir energies: $W=\sum_{i}r_{i}(\epsilon_{i}'-\epsilon_{i})$. Naturally, a two-level work reservoir can be treated as a special case where $W=\epsilon'-\epsilon$. Our work value definition is similar to that in stochastic thermodynamics: $dw = \sum_{i}p_{i}d\epsilon_{i}$, where the work is defined as the system energy change while keeping the system probability distribution unchanged \cite{aaberg2013truly}.

Here, though, the probability distribution components of the work reservoirs do not change overall. For each nonzero component, there is a corresponding energy level change in the work reservoir. Our results show that we can achieve reversibility in single-shot thermodynamics with multi-level work reservoirs. The price paid, however, is that the size of the thermal baths must be infinite. The dissipation can be written as:
\begin{align}
\Sigma= I(\rho_{S}';\rho_{B}') + D_{1}(\rho_{B}'||\tau_{B})
  ~,
\end{align}
where $I(\cdot;\cdot)$ is the mutual information \cite{reeb2014improved,richens2018finite}. Since the heat $Q$ transferred to the bath is nonzero, if we only have thermal baths of finite size, the dissipation is strictly positive. Appendix \ref{appenix:finiterealization} gives an example where we construct the joint unitary operator explicitly. We show that to approach zero dissipation, the bath size must be infinite.

References \cite{gemmer2015single,gallego2016thermodynamic,alhambra2018work} develop the general framework of work extraction in single-shot thermodynamics. 
Rather than considering strict energy conservation, work extraction can be monitored via average energy conservation \cite{skrzypczyk2014work}. There, work extraction uses a series of transformations, arriving at the same bound when the number of transformations diverges. Reference \cite{alhambra2016fluctuating} considers a weighted Hamiltonian $H_{W}=\int dx x \ra{x}\la{x}$ as a work reservoir. With translational invariance, it derives several compact fluctuation theorems. This allows changes in work reservoir probability distribution, but assumes the work reservoir energy levels are unbounded. References \cite{halpern2016beyond,halpern2018beyond} consider the work extraction of systems that exchange both energy and particles with the environment with multi-level batteries.

In contrast, our development here keeps the work reservoir probability distribution unchanged. This follows from the entropyless assumption of work reservoirs. Reference \cite{lipka2021second} considers a work reservoir with lower-bounded energy levels. Reference \cite{aaberg2018fully} systematically explores quantum fluctuation theorems. Recently, in single-shot thermodynamics, there are other setups that extract work  equal to the (negative) free energy difference \cite{lostaglio2015stochastic,muller2018correlating,sapienza2019correlations}. In this, correlations build up between catalysts and so stochastic independence of catalysts allows extracting more work from given states. 
    
Generalizing to multi-level work reservoirs offers several new directions in nanoscale thermodynamics. Since work is no longer deterministic, it is natural to ask how to compute higher moments $\avg{W^{n}}$ $(n>1)$ and to construct a fluctuation theorem for the work probability distribution. With two-level work reservoirs, the characteristic functions of work extraction and state formation are the R\'enyi $\alpha=0$ and $\alpha = \infty$ divergences, respectively. What are the characteristic functions of work extraction and state formation with multi-level work reservoirs? Our development focused on single-copy state transitions. The structure of the efficient work reservoirs for more complicated state transitions---for example, mapping an input information tape to output tape \cite{barnett2015computational}---must wait for the future.
    
Our development focused only on the net input-output mapping, without considering details of the stochastic map in between. The stochastic map connecting an input to an output here is not unique. If we only consider the work expectation value $\avg{W}$, the change in expectation value of energy in work reservoirs coincides with the expectation value of work in the \emph{two point measurement} (TPM) scheme commonly used in stochastic quantum thermodynamics \cite{esposito2009nonequilibrium}. For higher moments $\avg{W^{n}}$ $(n>1)$ in TPM, however, the values depend on the stochastic maps. Moreover, one cannot determine higher moments uniquely with only initial and final work reservoir states. We can also study the minimal cost of a stochastic map not only a specific state transition. References \cite{faist2015minimal,faist2018fundamental} explored the minimal cost of quantum channels with two-level work reservoirs. We leave the minimal work cost with multi-level work reservoirs also to the future. 

Along these lines, what if we allow coherence in both the system and the work reservoir? For example, what if $\rho_{W}=\sum_{ij}\rho_{ij} \ra{W_{i}}\la{W_{j}}$ and $\rho_{W}=\sum_{ij}\rho_{ij} \ra{W'_{i}}\la{W'_{j}}$? To address state transitions with coherence, $\alpha-$R\'enyi divergences are insufficient \cite{lostaglio2015quantum,lostaglio2015description}. Can we achieve the bounds set by free energy difference when the states are not block-diagonal in energy eigenstates with those work reservoirs? Again, we leave this open for the future efforts.

\section{Conclusion}

We generalized two-level work reservoirs commonly used in single-shot thermodynamics to multi-level work reservoirs and systematically analyzed arbitrarily-small-dissipation state transitions with the latter. We derived equivalent conditions for arbitrarily-small-dissipation transitions in single-shot thermodynamics: thermomajorization curve coincidence and $\alpha-$R\'enyi divergence equality. We showed that for any state transition, we can always construct a work reservoir to approach zero dissipation.

We also considered cases where the initial system Hamiltonian differs from the final Hamiltonian. The efficient work reservoir, though, for a specific state transition is not unique. For work extraction and state formation in this setting, we constructed the efficient work reservoir with minimal dimension. We showed that all thermomajorization curves at inverse temperature $\beta$ form a monoid and characterized all possible efficient reservoirs for work extraction and state formation. These allowed us to analyze nanoscale engines that employ efficient work reservoirs, demonstrating that they approach Carnot efficiency.


\section{Acknowledgements}

The authors thank Han Zhang, Paul Riechers, Ariadna Venegas-Li, David Gier, Hyun-Soo Kim, and Komal Sah for illuminating discussions, as well as the Telluride Science Research Center for its hospitality during visits and the participants of the Information Engines workshop there for their valuable feedback. This material is based on work supported by, or in part by, the U.S. Army Research Laboratory and U.S. Army Research Office under Grant No. W911NF-21-1-0048. A.B.B. acknowledges support from the Templeton World Charity Foundation Power of Information fellowships TWCF0337 and TWCF0560, from the Irish Research Council under grant number IRCLA/2022/3922, and from the Foundational Questions Institute and Fetzer Franklin Fund, a donor advised fund of the Silicon Valley Community Foundation, grant number FQXi-RFP-IPW-1910.

\appendix

\section{Free energy work bound}
\label{appendix:Free energy work bound}

Reference \cite{esposito2010entropy} establishes the Second Law of thermodynamics for the entropy production of a system $S$ in contact with a heat bath $B$ at temperature $T$:
\begin{align}
    \Sigma \equiv  Q/T + \Delta S(\rho) \geq 0.
\end{align}
Here, $Q$ is the average heat that was dissipated in the bath, $\Sigma$ is the total entropy production, and $S(\rho) \equiv - k_B \text{Tr} \left[ \rho \ln \rho \right]$ is the von Neumann entropy of the system $S$.  The resulting  bound on heat is:
\begin{align}
     Q/T \geq -\Delta S(\rho),
\end{align}
which is a quantum version of Landauer's principle \cite{alicki2004thermodynamics, reeb2014improved, riechers2021initial}.  We can bound the work by noting the First Law of thermodynamics: the change in average energy of the system is equal the minus the heat flow and work produced from the system:
\begin{align}
    \Delta \langle E \rangle = - Q-W,
\end{align}
where $\langle E \rangle = \text{Tr} \left[ \rho H \right]$. Applying the entropy bound on heat to the work production, we find the work production in transforming $\rho \rightarrow \rho'$ has the upper bound:
\begin{align}
    W & \leq T \Delta S(\rho) - \Delta \langle E \rangle
    \\ & = T(S(\rho')-S(\rho)) -(\text{Tr}\left[\rho' H\right]- \text{Tr}\left[\rho H\right]).
\end{align}
With the free energy defined:
\begin{align}
    F(\rho) \equiv \text{Tr}( \rho H) - T S(\rho),
\end{align}
we have an upper bound on work via the change in free energy:
\begin{align}
   W \leq F (\rho)-F(\rho').
\end{align}

Furthermore, we have a simplification when the Hamiltonian $H$ is the same for the initial and final state of the system. The Gibbs state $\tau$ of Hamiltonian $H$ obeys the relationship:
\begin{align}
    \tau = \frac{e^{-H/k_BT}}{\text{Tr} \left[e^{-H/k_BT} \right]},
\end{align}
which gives an inverse expression:
\begin{align}
    H= - k_B T \ln \tau -k_B T \ln \text{Tr} \left[e^{-H/k_BT} \right].
\end{align}
Plugging this into the average energy in the bound on work production, we obtain a change in relative entropies:
\begin{align}
    W \leq \left[ D_{1}(\rho|| \tau)-D_1(\rho'|| \tau) \right].
\end{align}
where $D_1(\rho||\sigma) \equiv \text{Tr} \left[ \rho \ln \rho- \rho \ln \sigma \right]$ is the quantum relative entropy.

\section{Thermal operations}
\label{appendix:thermaloperations}

Our results are based on the resource theory approach to quantum thermodynamics, several results from which we briefly note here. See Refs. \cite{lostaglio2019introductory, ng2019resource,gour2015resource} for more comprehensive reviews.

The central idea is to define a set of operations---the \emph{free operations}---and systematically analyze all possible state transitions under free operations. Suppose our state is $\rho_{S}$ with Hamiltonian $H_{S}$. The set of allowed transitions then contains all joint energy-preserving unitary $U$ operations between the system and a thermal bath with the Hamiltonian $H_{B}$ at inverse temperature $\beta$:
\begin{align}
    [U,H_{S}+H_{B}] & = 0
    ~,
\end{align}
followed by the partial trace over the thermal bath:
\begin{align}
    \rho'_{S} & = \mathcal{E}(\rho_{S}) \\
    & =\mathrm{Tr}_{B}\big(U (\rho_{S} \otimes \tau_{B}) U^{\dagger} \big)
    ~,
\end{align}
where $\tau_{B}=e^{-\beta H_{B}}/Z_{B}$ is the Gibbs state of the thermal bath. The maps $\mathcal{E}$ are called \emph{thermal operations}.

Suppose the eigenvalues of $\rho_{S}$ and $\rho_{S}'$ are $\{{p}_{i}\}_{i=1}^{n}$ and $\{p'_{i}\}_{i=1}^{n}$ and the associated energy levels are $\{e_{i}\}_{i=1}^{n}$. Such a transition is equivalent to there being a stochastic matrix $G$ such that $G\boldsymbol{p}=\boldsymbol{p'}$ and $G\boldsymbol{\tau}=\boldsymbol{\tau}$ \cite{horodecki2013fundamental}.

We can also use a geometric method to determine whether such a transition exists. A key concept is the \emph{thermomajorization curve} \cite{horodecki2013fundamental}. We first rank $\{{p}_{i}\}_{i=1}^{n}$ in descending order of ${p}_{i}e^{\beta e_{i}}$. This is called \emph{$\beta-$order}. The thermomajorization curve of a state $\rho_{S}$ is formed by connecting points:
\begin{align}
    \Big\{ \sum_{i=1}^{k} e^{-\beta {e}_{i}^{\downarrow}}, \sum_{i=1}^{k} {p}_{i}^{\downarrow} \Big\}_{k=1}^{n}
\end{align}
piecewise linearly where $\downarrow$ means that ${p}_{i}$ and ${e}_{i}$ have been $\beta-$ordered. If the thermomajorization curve of a state $\rho_{S}$ lies above or on the thermomajorization curve of another state $\rho_{S}'$, we say $\rho_{S}$ \emph{thermomajorizes} $\rho_{S}'$. The central result is that $\rho_{S}$ can be converted to  $\rho_{S}'$ through a thermal operation if and only if $\rho_{S}$ thermomajorizes $\rho_{S}'$. 

Next, we briefly review work extraction and the work of state formation. Consider a work reservoir that is a two-level system with Hamiltonian $H_{W}=W_{0}\ra{W_0}\la{W_0}+W_{1}\ra{W_1}\la{W_1}$. The task is to determine if the maximal work can be extracted from a state $\rho_{S}$. This is the maximal work change $W_{1}-W_{0}$ such that $\rho_{S} \otimes \ra{W_0}\la{W_0} \to \tau_{S} \otimes \ra{W_1}\la{W_1}$ is allowed by thermal operations. This is elegantly determined from the thermomajorization curve.

\begin{figure}[t]
    \centering
    \begin{tikzpicture}
    \draw[black, very thick](0,0) rectangle(6.5,4) node [label={[shift={(0.1,-4.7)}]:{\small $ Z_{S}Z_{W}$}}] {};
    \draw[blue, very thick] (0,0)--(1,3);
    \draw[blue, very thick] (1,3)--(3,3.8);
    \draw[blue, very thick] (3,3.8)--(4,4);
    \draw[blue, very thick] (4,4)--(6.5,4);
    \draw[black, thick, dashed] (4,4)--(4,0) node[anchor=north]{\small$a e^{-\beta W_{0}}$};
    \draw[red, thick] (0,0)--(5.3,4);
    \draw[red, thick] (5.3,4)--(6.5,4);
    \draw[black, thick, dashed] (5.3,4)--(5.3,0) node[anchor=north]{\small$e^{-\beta W_{1}} Z_{S}$};
    \draw (0,4) node[anchor=east] {1};
    \end{tikzpicture}
\caption{Deterministic work extraction: The blue curve is the thermomajorization curve of $\rho \otimes \ra{0}\la{0}$. $a\leq Z_{S}$ is the $x$-coordinate of the point where the thermomajorization curve of $\rho_{S}$ reaches $1$. The red curve is the thermomajorization curve of $\tau_{S}\otimes \ra{1}\la{1}$. $Z_{S}$ and $Z_{W}$ are partition functions of the system and the work reservoir, respectively.}
\label{fig:workextraction}
\end{figure}

For the initial curve to thermomajorize the final curve, we must have $a e^{-\beta W_{0}} \leq Z_{S} e^{-\beta W_{1}}$. See Fig. \ref{fig:workextraction}. Here, $a$ is related to R\'enyi divergence via: $D_{0}(\rho_{S} \| \tau_{S})=-\log(a/Z_{S})$. We have the bound $W_{1}-W_{0}\leq \kt D_{0}(\rho_{S} \| \tau_{S})$. The equal sign holds when two curves reach the height $1$ at the same point.

Similarly, we can consider the reverse question: What is the minimal work needed to form state $\rho_{S}$? Or, in other words, what is the minimal $W_{1}-W_{0}$ such that $(\tau_{S} \otimes \ra{W_1}\la{W_1},H_{S}+H_{W}) \to (\rho_{S} \otimes \ra{W_0}\la{W_0},H_{S}+H_{W})$ is allowed by thermal operations?

\begin{figure}[t]
    \centering
    \begin{tikzpicture}
    \draw[black, very thick](0,0) rectangle(6.5,4) node [label={[shift={(0.1,-4.7)}]:{\small $ Z_{S}Z_{W}$}}] {};
    \draw[blue, very thick] (0,0)--(2,3);
    \draw[blue, very thick, dashed] (2,3)--(2.66,4);
    \draw[blue, very thick] (2,3)--(4,3.8);
    \draw[blue, very thick] (4,3.8)--(5,4);
    \draw[blue, very thick] (5,4)--(6.5,4);
    \draw[red, thick] (0,0)--(1,4);
    \draw[red, thick] (1,4)--(6.5,4);
    \draw[black, thick, dashed] (1,4)--(1,0) node[anchor=north]{\small$Z_{S} e^{-\beta W_{1}}$};
    \draw (0,4) node[anchor=east] {1};
    \end{tikzpicture}
\caption{Deterministic work of state formation: The blue curve is the thermomajorization curve of $\rho_{S} \otimes \ra{W_0}\la{W_0}$. The red curve is the thermomajorization curve of $\tau_{S}\otimes \ra{W_1}\la{W_1}$.}
\label{fig:workformation}
\end{figure}

For the initial curve to thermomajorize the final curve, the slope of the on-ramp part of the initial curve must not be less than the largest slope in the final curve:
\begin{align}
    \frac{1}{Z_{S}} e^{\beta W_{1}}\geq e^{\beta W_{0}} \max_{i} \frac{p_{i}}{e^{-\beta \epsilon_{i}}}
    ~.
\end{align}
And:
\begin{align}
    \max_{i} \left\{ \frac{p_{i}}{e^{-\beta \epsilon_{i}}} \right\} = D_{\infty}(\rho_{S}\|\tau_{S})
    ~.
\end{align}
Giving:
\begin{align}
    W_{1}-W_{0}\geq \kt D_{\infty}(\rho_{S}\|\tau_{S})
    ~.
\end{align}

If there exists an auxiliary system---a \emph{catalyst}---with Hamiltonian $H_{C}$ and state $\rho_{C}$ such that the transition $(\rho_{S}\otimes \rho_{C}, H_{S}+H_{C}) \to (\rho_{S}'\otimes \rho_{C}, H_{S}+H_{C})$ is possible, we say the transition $(\rho_{S}, H_{S}) \to (\rho_{S}', H_{S})$ can be achieved by a \emph{catalytic thermal operation}.

The criterion of the catalytic thermomajorization is given in terms of R\'enyi $\alpha$-divergences. There exists a transition $(\rho_{S}, H_{S}) \overset{CTO}{\longrightarrow} (\rho_{S}', H_{S})$ if and only if \cite{brandao2015second}:
\begin{align}
    D_{\alpha}(\rho_{S}||\tau_{S}) \geq D_{\alpha}(\rho_{S}'||\tau_{S})
    ~,
\end{align}
for all $\alpha \in \mathbb{R}$. If we are allowed to invest an infinitesimal amount of work, only $\alpha\geq 0$ is needed.

We can also study work extraction and state formation in two-level work reservoirs with the help of catalysts. For work extraction:
\begin{align}
(\rho_{S} & \otimes |W_0\rangle \langle W_0|, H_{S}+H_{W}) \nonumber \\
    & \quad \to (\tau_{S} \otimes |W_1\rangle \langle W_1|, H_{S}+H_{W})
    , 
\end{align}
we must have:
\begin{align}
    D_{\alpha}(\rho_{S} & ||\tau_{S}) + D_{\alpha}(\ra{W_0}\la{W_0}||\tau_{W}) \nonumber \\
    & ~  \geq D_{\alpha}(\tau_{S}||\tau_{S})  + D_{\alpha}(\ra{W_1}\la{W_1}||\tau_{W})
    ~.
\end{align}
Giving:
\begin{align}
    W_{1}-W_{0}\leq \kt D_{\alpha}(\rho_{S}||\tau_{S})
    ~,
\end{align}
for all $\alpha\geq 0$.
So, we have:
\begin{align}
    W_{1}-W_{0} & \leq \inf_{\alpha\geq 0}\kt D_{\alpha}(\rho_{S}||\tau_{S}) \nonumber \\
    & = \kt D_{0}(\rho_{S}||\tau_{S})
    ~.
\end{align}

For state formation:
\begin{align}
(\tau_{S} \otimes |W_1\rangle & \langle W_1|, H_{S}+H_{W}) \nonumber \\
    & \to (\rho_{S} \otimes |W_0\rangle \langle W_0|, H_{S}+H_{W})
    ~,
\end{align}
we must have:
\begin{align}
D_{\alpha}(\tau_{S}||\tau_{S}) & + D_{\alpha}(\ra{W_1}\la{W_1}||\tau_{W}) \nonumber \\
    & \geq D_{\alpha}(\rho_{S}||\tau_{S}) + D_{\alpha}(\ra{W_0}\la{W_0}||\tau_{W})
  ~.
\end{align}
Giving:
\begin{align}
    W_{1}-W_{0}\geq \kt D_{\alpha}(\rho_{S}||\tau_{S})
    ~,
\end{align}
for all $\alpha\geq 0$. So, we have:
\begin{align}
W_{1}-W_{0} & \geq \sup_{\alpha\geq0}\kt D_{\alpha}(\rho_{S}||\tau_{S}) \nonumber \\
   & =\kt D_{\infty}(\rho_{S}||\tau_{S})
    ~.
\end{align}

\section{Proofs}\label{appendix:proof}

\subsection{Proof of Theorem \ref{thm:coincide}}

We first list the precise statement on the connection between the thermomajorization curves and existence of the thermal operations and Gibbs preserving stochastic matrices and then list a theorem regarding to thermomajorization curve coincide. After that, we prove Theorem \ref{thm:coincide}.

The distance we use is norm-$1$ distance:
\begin{align}
    \|\rho -\sigma \|_{1}=\mathrm {Tr} \left({\sqrt {(\rho -\sigma )^{\dagger }(\rho -\sigma )}}\right)~.
\end{align}
Since we only consider diagonal states, the norm-1 distance is simply:
\begin{align}
    \|\rho -\sigma \|_{1}=\sum_{i}|(\boldsymbol{p}_{\rho})_{i} - (\boldsymbol{p}_{\sigma})_{i}|~.
\end{align}

\begin{theorem}[Thermal Nielsen’s theorem]\label{thm:epsilonthermoandcurves}
    Consider two block diagonal states $\rho$ and $\sigma$ with Hamiltonian $H$ and their corresponding population vectors are $\boldsymbol{p}_{\rho}$ and $\boldsymbol{p}_{\sigma}$. 
    \begin{enumerate}
    \item For any $\epsilon>0$, there exists a thermal operation $\mathcal{E}$ such that $\mathcal{E}(\rho)$ is arbitrarily close to $\sigma$, i.e., $||\mathcal{E}(\rho)-\sigma||_{1}<\epsilon$ if and only if the thermomajorization curve of $\rho$ lies above or on the thermomajorization curve of $\sigma$.
    \item There exists a Gibbs preserving stochastic map $G$ such that $G\cdot\boldsymbol{p}_{\rho}= \boldsymbol{p}_{\sigma}$ if and only if the thermomajorization curve of $\rho$ lies above or on the thermomajorization curve of $\sigma$.
    \end{enumerate}
\end{theorem}

\begin{proof}
For the proof, see Theorems 6 and 7 and Remark 10 in \cite{lostaglio2019introductory}.
\end{proof}

Theorem \ref{thm:epsilonthermoandcurves} shows whether the existence of quantum thermal operations or Gibbs preserving stochastic matrices is related to thermomajorization curves. Next, we list a theorem related to thermomajorization coincidence.

\begin{theorem}\label{classicalcoincide}
Consider two states $\rho$ and $\sigma$ with Hamiltonian $H$. If the thermomajorization curve of $\rho$ lies above or on the thermomajorization curve of $\sigma$, then $D(\rho||\tau)\geq D(\sigma||\tau)$. The equality signs hold if and only if two curves coincide.
\end{theorem}

\begin{proof}
    Suppose the population vectors of state $\rho$ and $\sigma$ are $\boldsymbol{p}_{\rho}$ and $\boldsymbol{p}_{\sigma}$. The thermomajorization curve of $\rho$ lies above and on the thermomajorization curve of $\sigma$. From Theorem \ref{thm:epsilonthermoandcurves} there exists a Gibbs preserving stochastic matrix $G$ such that $G\cdot\boldsymbol{p}_{\rho}=\boldsymbol{p}_{\sigma}$. Since $\rho$ and $\sigma$ are block-diagonal, the relative entropy is the same as its classical version:
    \begin{align}
        D(\rho||\tau) &= D(\boldsymbol{p}_{\rho}||\boldsymbol{p}_{\tau})\\
        D(\sigma||\tau) &= D(\boldsymbol{p}_{\sigma}||\boldsymbol{p}_{\tau})~.
    \end{align}
    From data processing inequality, we have:
    \begin{align}
        D(\boldsymbol{p}_{\rho}||\boldsymbol{p}_{\tau})&\geq D(G\cdot\boldsymbol{p}_{\rho}||G\cdot\boldsymbol{p}_{\tau})  \\
        &= D(\boldsymbol{p}_{\sigma}||\boldsymbol{p}_{\tau})~.
    \end{align}
    This completes the first part of proof. 
    
    The data processing inequality saturates if and only if there exists a recovery map $R$ defined by $R_{ij}=G_{ji} (\boldsymbol{p}{_{\tau}})_{i}/(\boldsymbol{p}{_{\tau}})_{j}$ such that $R\cdot\boldsymbol{p}_{\sigma}=\boldsymbol{p}_{\rho}$, where $(\cdot)_{ij}$ is the $ij$ component of the matrix \cite{wilde2013quantum}. It is straightforward to show that $R$ preserves the Gibbs distribution: $R\cdot \boldsymbol{p}{_{\tau}}=\boldsymbol{p}{_{\tau}}$. So, $\boldsymbol{p}{_{\sigma}}$ thermomajorizes $\boldsymbol{p}{_{\rho}}$. $\sigma$'s thermomajorization curve lies above or on $\rho$'s thermomajorization. Hence, their thermomajorization curves coincide.
\end{proof}
Now, we write down the precise version of Theorem \ref{thm:coincide}:

\begin{theorem}\label{thm:precisecoincide}
    Consider two $d-$dimension diagonal states $\rho$ and $\sigma$. The following two are equivalent:
    \begin{enumerate}
        \item \label{coincideA1} The thermomajorization curves of states ${\rho}$ and ${\sigma}$ coincide.
        \item \label{coincideA2} For all $\epsilon_{1},\epsilon_{2}>0$, there exists a thermal operation $\mathcal{E}$ such that $||\mathcal{E}(\rho)-\sigma||_{1}<\epsilon_{1}$ and the corresponding entropy production $\Sigma_{\rho \to \mathcal{E}(\rho)}<\epsilon_{2}$~.
    \end{enumerate} 
\end{theorem}

\begin{proof}
\ref{coincideA1} $\to$ \ref{coincideA2}: 
Since the thermomajorization curves of state $\rho$ and $\sigma$ coincide, for any $\epsilon$ there exists a thermal operation $\mathcal{E}$ such that $||\mathcal{E}(\rho)-\sigma||_{1} <\epsilon$. The upper bound of entropy production is given as follows. The thermomajorization curves of $\rho$ and $\sigma$ coincide. From Theorem \ref{classicalcoincide}, we have $D_{1}(\rho||\tau)=D_{1}(\sigma||\tau)$. By definition of the relative entropy, $D_{1}(\mathcal{E}(\rho)||\tau)=S(\mathcal{E}(\rho)) - \beta  \mathrm{Tr}(\mathcal{E}(\rho) H)$. $\mathcal{E}(\rho)$ and $\sigma$ are $\epsilon$ close. From Zhang–Audenaert inequality \cite{wilde2013quantum}, we have:
\begin{align}
    |S(\sigma)-S(\mathcal{E}(\rho))|\leq \frac{1}{2} \epsilon (\log d - 1) + H(\epsilon)~,
\end{align}
where $H(\cdot)$ is the binary entropy function. This gives the entropy difference upper bound. Second term in relative entropy is bounded by:
\begin{align}
    |\mathrm{Tr}((\sigma-\mathcal{E}(\rho))H)|\leq \epsilon E_{\max}~,
\end{align}
    where $E_{\max}$ is the maximal eigenvalues in the Hamiltonian $H$. The relative entropy is bounded by:
\begin{align}
    |D_{1}(\sigma||\tau) & - D_{1}(\mathcal{E}(\rho)||\tau)| \nonumber \\
    & \quad \leq \frac{1}{2} \epsilon (\log d - 1) + \epsilon \beta E_{\max} + H(\epsilon)~.
\end{align}
    Since we have $D_{1}(\rho||\tau)=D_{1}(\sigma||\tau)$:
\begin{align}
    |D_{1}(\rho||\tau) & - D_{1}(\mathcal{E}(\rho)||\tau)| \nonumber \\
    & \quad \leq \frac{1}{2} \epsilon (\log d - 1) + \epsilon \beta E_{\max} + H(\epsilon)~.
\end{align}
    Let $f(\epsilon)=\frac{1}{2} \epsilon (\log d - 1) + \epsilon \beta E_{\max} + H(\epsilon)$ is an increasing function about $\epsilon$ in $[0,\frac{1}{2}]$ and $f(0)=0$, $f(1/2)= \frac{1}{4}(\log d - 1) + \frac{1}{2} \beta E_{\max} + \log 2$. We denote the corresponding inverse function in $[0, f(1/2)]$ as $f^{-1}(x)$.  For any $\epsilon_{1},\epsilon_{2}>0$, if $\epsilon_{2}<f(1/2)$, we can take $\epsilon=\frac{1}{2}\min \{\epsilon_{1}, f^{-1}(\epsilon_{2})\} $. Then $||\mathcal{E}(\rho)-\sigma||_{1} <\epsilon \leq \frac{1}{2}\epsilon_{1}<\epsilon_{1}$ and $|D_{1}(\rho||\tau)-D_{1}(\mathcal{E}(\rho)||\tau)|\leq f(\epsilon)\leq f(\frac{1}{2}f^{-1}(\epsilon_{2}))<f(f^{-1}(\epsilon_{2}))=\epsilon_{2}$. If $\epsilon_{2}\geq f(1/2)$, we take $\epsilon=\min \{\epsilon_{1}, \frac{1}{4} \}$. Then $||T(\rho)-\sigma||_{1} <\epsilon \leq \epsilon_{1}$ and $|D_{1}(\rho||\tau)-D_{1}(T(\rho)||\tau)|\leq f(\epsilon)\leq f(1/4)<f(1/2)\leq \epsilon_{2}$.

\ref{coincideA2} $\to$ \ref{coincideA1} by contradiction: Since for all $\epsilon_{1}>0$, there exists a thermal operation $\mathcal{E}$ such that $||\mathcal{E}(\rho)-\sigma||_{1}<\epsilon_{1}$, the thermomajorization curve $\rho$ lies above or on the the thermomajorization curve $\sigma$ (Theorem \ref{thm:epsilonthermoandcurves}). Assume the thermomajorization curves of $\rho$ and $\sigma$ do not coincide, then $|D_{1}(\rho||\tau)-D_{1}(\sigma||\tau)|\neq0$ (Theorem \ref{classicalcoincide}).  We give a bound on $|D_{1}(\rho||\tau)-D_{1}(\sigma||\tau)|$:
    \begin{align}
        |& D_{1}(\rho||\tau) - D_{1}(\sigma||\tau)| \nonumber \\
        &\leq|D_{1}(\rho||\tau)-D_{1}(\mathcal{E}(\rho)||\tau)|
        +|D_{1}(\mathcal{E}(\rho)||\tau)-D_{1}(\sigma||\tau)| \nonumber \\
        &<\epsilon_{2} + \frac{1}{2} \epsilon_{1} (\log d - 1) 
        + \epsilon_{1} \beta E_{\max} + H(\epsilon_{1})~.
    \end{align}
    Since $\epsilon_{1},\epsilon_{2}$ are arbitrary and:
    \begin{align}\label{eqn:proofimit}
        \lim_{\epsilon_{1},\epsilon_{2}\to 0 }\epsilon_{2} + 2 \epsilon_{1} (\log d - 1) + \epsilon_{1} \beta E_{\max} + H(\epsilon_{1})=0~.
    \end{align}
    We know $|D_{1}(\rho||\tau)-D_{1}(\sigma||\tau)|\neq0$ and is a finite fixed positive number. This contradicts with Eq. \eqref{eqn:proofimit}. So the two curves must coincide. 
\end{proof}

\subsection{Proof to Theorem \ref{thm: catalystwith0entropyproduction}}

We first write down the precise version of Theorem \ref{thm: catalystwith0entropyproduction}.
\begin{theorem}
    Consider two states $\rho$ and $\sigma$ with Hamiltonian $H$ and a catalyst state $c$ with Hamiltonian $H_{c}$. For any $\epsilon_{1},\epsilon_{2}>0$, if there exists a thermal operation $\mathcal{E}$ such that $||\mathcal{E}(\rho\otimes c)-\sigma\otimes c||_{1}<\epsilon_{1}$ and the corresponding entropy production $\Sigma_{\rho\otimes c\to \sigma\otimes c}<\epsilon_{2}$, then there exists another thermal operation $\mathcal{T}$ such that $||\mathcal{T}(\rho)-\sigma||_{1}<\epsilon_{1}$ and the corresponding entropy production $\Sigma_{\rho\to \sigma}<\epsilon_{2}$.
\end{theorem}

\begin{proof}
From Theorem \ref{thm:precisecoincide}, the thermomajorization curves of $\rho \otimes c$ and $\sigma \otimes c$ coincide. Next, we show that the thermomajorization curves of $\rho$ and $\sigma$ coincide. Suppose $\boldsymbol{f}_{\rho, H}=\{(y^{(\rho)}_{i},k^{(\rho)}_{i})\}_{i}$, $\boldsymbol{f}_{\sigma,H}=\{ (y^{(\sigma)}_{i},k^{(\sigma)}_{i})\}_{i}$, and $\boldsymbol{f}_{c, H_{c}}=\{(y^{(c)}_{i},k^{(c)}_{i})\}_{i}$, respectively. Here, we coarse grain all segments with the same slopes and there are no repetitive slopes in $\boldsymbol{f}_{\rho, H}$, $\boldsymbol{f}_{\sigma,H}$, and $\boldsymbol{f}_{c, H_{c}}$. The largest slope of the $\rho \otimes c $ curve is $k^{(\rho)}_{1} \cdot k^{(c)}_{1}$ with $y-$coordinate change $y^{(\rho)}_{1} \cdot y^{(c)}_{1}$. And, the largest slope of the $\sigma \otimes c $ curve is $k^{(\sigma)}_{1} \cdot k^{(c)}_{1}$ with $y$ coordinate change $y^{(\sigma)}_{1} \cdot y^{(c)}_{1}$. Since the curves of $\rho \otimes c $ and $\sigma \otimes c$ coincide, we must have:
\begin{align}
        k^{(\rho)}_{1} \cdot k^{(c)}_{1} &= k^{(\sigma)}_{1} \cdot k^{(c)}_{1} \\
        y^{(\rho)}_{1} \cdot y^{(c)}_{1} &= y^{(\sigma)}_{1} \cdot y^{(c)}_{1}
        ~.
\end{align}
This leads to $k^{(\rho)}_{1}=k^{(\sigma)}_{1}$ and $y^{(\rho)}_{1}=y^{(\sigma)}_{1}$.

We can remove the contribution of  $(k^{(\rho)}_{1}, y^{(\rho)}_{1})$ and $(k^{(\sigma)}_{1}, y^{(\sigma)}_{1})$ from the curves $\rho \otimes c$ and $\sigma \otimes c$, respectively. The two new curves also coincide since we remove identical segments from two identical thermomajorization curves. With the two new curves and the similar argument, we have:
\begin{align}
        k^{(\rho)}_{2} \cdot k^{(c)}_{1} &= k^{(\sigma)}_{2} \cdot k^{(c)}_{1} \\
        y^{(\rho)}_{2} \cdot y^{(c)}_{1} &= y^{(\sigma)}_{2} \cdot y^{(c)}_{1}
        ~,
\end{align}
which lead to $k^{(\rho)}_{2}=k^{(\sigma)}_{2}$ and $y^{(\rho)}_{2}=y^{(\sigma)}_{2}$. If we continue this procedure, we can show that $k^{(\rho)}_{i}=k^{(\sigma)}_{i}$ and $y^{(\rho)}_{i}=y^{(\sigma)}_{i}$ for any $i$. Then the $\rho$ and $\sigma$ curves coincide. So for any $\epsilon_{1},\epsilon_{2}$, there exists another thermal operation $\mathcal{T}$ such that $||\mathcal{T}(\rho)-\sigma||_{1}<\epsilon_{1}$ and the corresponding entropy production $\Sigma_{\rho\to \sigma}<\epsilon_{2}$.
\end{proof}

\subsection{Proof to Theorem \ref{thm:coincide2}}
We first show a theorem regarding to equal of $\alpha$- R\'enyi entropy.

\begin{theorem}\label{thm:equalalpharenyiapp}
    Consider $\boldsymbol{p},\boldsymbol{q}$ two $m-$dimension probability distributions. If $D_{\alpha}(\boldsymbol{p}||\boldsymbol{\eta})=D_{\alpha}(\boldsymbol{q}||\boldsymbol{\eta})$ for any $\alpha\in \mathbf{R}$ where $\boldsymbol{\eta}$ is the $m-$dimension uniform distribution. Then $p,q$ are same up to a reorder.
\end{theorem}

\begin{proof}
    The $\alpha-$R\'enyi divergence of $\boldsymbol{p}$ from the uniform distribution $\boldsymbol{\eta}$ is:
    \begin{align}
       D_{\alpha}(\boldsymbol{p}||\boldsymbol{\eta})= \frac{1}{\alpha-1}  \log (||\boldsymbol{p}||_{\alpha})^{\alpha} + \log m~,
       \end{align}
    where $||\cdot||_{\alpha}$ is the $\alpha-$norm. And the $\infty$-R\'enyi divergence picks the maximal component in the distribution:
    \begin{align}
       D_{\infty}(\boldsymbol{p}||\boldsymbol{\eta}) = \max_{p_{i}} \log p_{i}+\log m~. 
    \end{align}
    The equal $\alpha-$R\'enyi means $\boldsymbol{p}$ and $\boldsymbol{q}$ have the same $\alpha-$norm. Taking $\alpha\to\infty$ gives:
    \begin{align}
      \max_{p_{i}} p_{i} = \max_{q_{i}} q_{i} ~.
    \end{align}
     $\boldsymbol{p}$ and $\boldsymbol{q}$ have the same maximal component. We can remove the corresponding maximal component from both distributions and they still have the same $\alpha-$norm:
     \begin{align}
         ||\boldsymbol{p} \setminus \{\max_{p_{i}} p_{i}\}||_{\alpha} = ||\boldsymbol{q} \setminus \{\max_{q_{i}} q_{i}\}||_{\alpha}~.
     \end{align}
     Again, we take $\alpha \to \infty$ which gives the second maximal components in $\boldsymbol{p}$ and $\boldsymbol{q}$ are same. Continuing this procedure leads that $\boldsymbol{p}$ and $\boldsymbol{q}$ are same up to a reorder.
\end{proof}

Now, we are ready to prove Theorem \ref{thm:coincide2}.
\begin{theorem}[Theorem \ref{thm:coincide2} ]\label{thm:coincide2app}
Given a system with Hamiltonian $H$ and states $\rho$ and $\sigma$, the following are equivalent: 
\begin{enumerate}
    \item \label{Coincide21a} Thermomajorization curves of $\rho$ and $\sigma$ coincide.
    \item \label{Coincide22a} $D_{\alpha}(\rho \| \tau) = D_{\alpha}(\sigma \| \tau),\  \text{~for~all~} \alpha \in \mathbb{R}$.
\end{enumerate}
\begin{proof}
\ref{Coincide21a} $\to$ \ref{Coincide22a}: Since the thermomajorization curves of $\rho$ and $\sigma$ coincide. There exists two Gibbs-preserving stochastic map $E$ and $G$ such that $E \cdot \boldsymbol{p}_{\rho} = \boldsymbol{p}_{\sigma}$ and $G \cdot \boldsymbol{p}_{\sigma} = \boldsymbol{p}_{\rho}$. With the data processing inequality of R\'enyi $\alpha-$divergence, for all $\alpha\in \mathbb{R}$, we have \cite{van2014renyi,brandao2015second}:
\begin{align}
    D_{\alpha}(\boldsymbol{p}_\rho \| \boldsymbol{p}_\tau) & \geq D_{\alpha}(E\boldsymbol{p}_\rho \| E\boldsymbol{p}_\tau)= D_{\alpha}(\boldsymbol{p}_\sigma \| \boldsymbol{p}_\tau) \\
    D_{\alpha}(\boldsymbol{p}_\sigma \| \boldsymbol{p}_\tau) & \geq D_{\alpha}(G\boldsymbol{p}_\sigma \| G\boldsymbol{p}_{\tau})= D_{\alpha}(\boldsymbol{p}_\rho \| \boldsymbol{p}_\tau)~.
\end{align}
Then $D_{\alpha}(\rho \| \tau) = D_{\alpha}(\sigma \| \tau)$, for all $\alpha \in \mathbb{R}$.
    
\ref{Coincide22a} $\to$ \ref{Coincide21a}: We use a basic tool in single-shot thermodynamics---the embedding map \cite{brandao2015second}. Here, the embedding map $\Gamma$ sends one distribution to a larger dimension distribution. And, $\Gamma$ maps the Gibbs distribution of the system to a larger uniform distribution (To avoid some technicalities, we assume the Gibbs distribution is rational). $\Gamma$ has the following properties \cite{brandao2015second}:
\begin{enumerate}
    \item $D_{\alpha}(\boldsymbol{p}_\sigma \| \boldsymbol{p}_\tau)= D_{\alpha}(\Gamma(\boldsymbol{p}_\sigma) \| \Gamma(\boldsymbol{p}_\tau))~.$
    \item $\boldsymbol{p}$ thermomajorizes $\boldsymbol{q}$ with respect to the Gibbs distribution if and only if $\Gamma(\boldsymbol{p})$ majorizes $\Gamma(\boldsymbol{q})$.
\end{enumerate}

From $D_{\alpha}(\rho \| \tau) = D_{\alpha}(\sigma \| \tau)$, we have
\begin{align}
    D_{\alpha}(\Gamma(\boldsymbol{p}_\rho) \| \Gamma(\boldsymbol{p}_\tau)) = D_{\alpha}(\Gamma(\boldsymbol{p}_\sigma) \| \Gamma(\boldsymbol{p}_\tau))
\end{align}
for all $\alpha$ and $\Gamma(\boldsymbol{p}_\tau)$ is a uniform distribution. From Theorem. \ref{thm:equalalpharenyiapp}, 
\begin{align}
    \Gamma(\boldsymbol{p}_\rho)=\Gamma(\boldsymbol{p}_\sigma)
\end{align}
up to a permutation. Then $\Gamma(\boldsymbol{p}_\rho)$ and $\Gamma(\boldsymbol{p}_\sigma)$ majorize each other. We have $\boldsymbol{p}_\rho$ and $\boldsymbol{p}_\sigma$ thermomajorize each other. Hence the thermomajorization curves of $\rho$ and $\sigma$ coincide.
\end{proof}
\end{theorem}

\section{Details on the constructions for any state transitions}\label{appendix:detailconstruction}

Proving this requires constructing the efficient work reservoir for $(\boldsymbol{p}_{S},\boldsymbol{E}_{S})\to(\boldsymbol{p}'_{S},\boldsymbol{E}_{S})$. We denote initial and final cumulative probability distributions of the system as $\boldsymbol{P}=\{{P}_{i}\}_{i\in\{0\}\cup \mathcal{S}}$ and $\boldsymbol{P}'=\{{P}_{i}'\}_{i\in\{0\}\cup \mathcal{S}}$, where ${P}_{0}={P}'_{0}=0$. And, they satisfy ${P}_{i}-{P}_{i-1}={p}_{i}$ and ${P}_{i}'-{P}_{i-1}'={p}_{i}'$ for all $i\in \mathcal{S}$. Let $\boldsymbol{R}=\{R_{i}\}_{i\in \{0\}\cup \mathcal{W}}=\boldsymbol{P} \cup \boldsymbol{P}'$---a cumulative probability distribution where $\mathcal{W}=\{1,2,\cdots,N\}$ and $N$ is the dimension of corresponding probability distribution, denoted $\boldsymbol{r}=\{r_{i}\}_{i\in\mathcal{W}}$. Then there exist mappings $\lambda,\lambda':\mathcal{W} \rightarrow \mathcal{S}$ from $\mathcal{W}=\{1,2, \cdots , N\}$ to system eigenstates $\mathcal{S}=\{1,2, \cdots n\}$ such that: 
\begin{align}
        p'_i & =\sum_{j \in \lambda^{-1}(i) }r_j \label{initialworkresprob}
        \\p_i & =\sum_{j \in \lambda'^{-1}(i) }r_j ~.\label{finalworkresprob}
\end{align}
Appendix \ref{appendix:lambdaconstruction} constructs the mappings $\lambda$ and $\lambda'$.

The work reservoir probabilities are $\boldsymbol{p}_{W}$ and $\boldsymbol{p}'_{W}$, where $\boldsymbol{p}_{W}=(\boldsymbol{r},\boldsymbol{0})$ and $\boldsymbol{p}_{W}'=(\boldsymbol{0},\boldsymbol{r})$. And the energy levels are $\boldsymbol{E}_W=\{\epsilon_{1}, \cdots, \epsilon_{N}, \epsilon_{1}', \cdots, \epsilon_{N}'\}$. To make this efficient for a $n-$dimensional transition $\boldsymbol{p}_{S} \to \boldsymbol{p}_{S}'$ in a system with energy levels $\boldsymbol{E}_S = \{e_1,\cdots ,e_n\}$, we require that:
\begin{enumerate}[label=(\alph*)]
 \item \label{workresslope}There exist sets of positive numbers $\{k_i\}_{i=1}^n$ and $\{k'_i\}_{i=1}^n$ such that:
    \begin{align}
        r_j e^{\beta \epsilon_j}=k_i, \text{ for all } j \in \lambda^{-1}(i)~,
        \label{initialworkresslope}
    \end{align}
    and:
    \begin{align}
        r_j e^{\beta \epsilon'_j}=k'_i, \text{ for all } j \in \lambda'^{-1}(i)~.
        \label{finalworkresslope}
    \end{align}
\item \label{sameslope} And:
\begin{align}
p_ie^{\beta e_{i}} \cdot k_{j}  = p_{j}'e^{\beta e_{j}} \cdot k_{i}', \text{for any pair } (i,j)~.
\end{align}
\end{enumerate} 

According to Theorem \ref{thm:coincide}, zero entropy is produced if and only if the thermomajorization curves of $\boldsymbol{p}_{SW}= \{p_ir_j\}_{i,j}$ over the energy levels $\boldsymbol{E}_{SW}= \{ e_i+\epsilon_j\}_{i,j}$, and $\boldsymbol{p}'_{SW}= \{p'_ir_j\}_{i,j}$ over the energy levels $\boldsymbol{E}'_{SW}= \{ e_i+\epsilon'_j\}_{i,j}$ are the same. (We neglect contributions from zero components in probability distribution.) From Eqs. \eqref{initialworkresprob} and \eqref{initialworkresslope}, the thermomajorization curve $f_{\boldsymbol{p}_{W},\boldsymbol{E}_{W}}$ has at most $n$ distinct slopes $\{ k_{i}\}_{i=1}^{n}$ with corresponding $y-$coordinate change $\{p_{i}'\}_{i=1}^{n}$; i.e.,  $\boldsymbol{f}_{\boldsymbol{p}_{W},\boldsymbol{E}_{W}}=\{(p_{i}', k_{i})\}_{i=1}^{n}$. For the system, we have $\boldsymbol{f}_{\boldsymbol{p}_{S},\boldsymbol{E}_{S}}=\{(p_{i}, p_{i}e^{\beta e_{i}})\}_{i=1}^{n}$. And, so, from Eq. \eqref{eqn:compositethermalcurve} we have  $\boldsymbol{f}_{\boldsymbol{p}_{SW},\boldsymbol{E}_{SW}}=\{p_{i}p_{j}',p_{i}e^{\beta e_{i}} k_{j}\}_{i,j=1}^{n}$. Similarly, we have $\boldsymbol{f}_{\boldsymbol{p}'_{SW},\boldsymbol{E}_{SW}}=\{p_{i}'p_{j},p'_{i}e^{\beta e_{i}} k'_{j}\}_{i,j=1}^{n}=\{p_{i}p_{j}',p'_{j}e^{\beta e_{j}} k'_{i}\}_{i,j=1}^{n}$. From Condition \ref{sameslope}, $\boldsymbol{f}_{\boldsymbol{p}_{SW},\boldsymbol{E}_{SW}} = \boldsymbol{f}_{\boldsymbol{p}'_{SW},\boldsymbol{E}_{SW}}$. That is, the two thermomajorization curves coincide. 

Next, we determine the energy levels $\{\epsilon_{1}, \cdots, \epsilon_{N}\}$ and $\{\epsilon_{1}', \cdots, \epsilon_{N}'\}$ explicitly. We fix one energy level, for example $\epsilon_{1}$, and express all other energy levels in terms of it. To determine $k_{i}$, from Condition \ref{sameslope} we have:
\begin{align}
    p_{j} e^{\beta e_{j}} k_{i} &= p_{i}' e^{\beta e_{i}} k_{j}' \\ 
    p_{j} e^{\beta e_{j}} k_{1} &= p_{1}' e^{\beta e_{1}} k_{j}'~.
\end{align}
Dividing gives:
\begin{align}
    k_{i}= k_{1} \frac{p'_{i}e^{\beta e_{i}}}{p'_{1}e^{\beta e_{1}}}
    ~,
\end{align}
from which we have:
\begin{align}\label{eqn:initiallevels}
    \epsilon_{x} = \epsilon_{1} + \kt \log \Bigg( \frac{r_{1}}{r_{x}}\frac{p_{i}'e^{\beta e_{i}}}{p_{1}'e^{\beta e_{1}}} \Bigg)
    ~,
\end{align}
for all $x \in \lambda^{-1}(i)$.
$k_{i}'$ can be determined through Condition \ref{sameslope} by setting $j=1$:
\begin{align}
    k_{i}'= k_{1}\frac{p_{i}e^{\beta e_{i}}}{p'_{1}e^{\beta e_{1}}}
    ~.
\end{align}
From which we have:
\begin{align}\label{eqn:finallevels}
    \epsilon'_{x} = \epsilon_{1} + \kt \log \Bigg( \frac{r_{1}}{r_{x}} \frac{p_{i}e^{\beta e_{i}}}{p'_{1}e^{\beta e_{1}}} \Bigg)
    ~,
\end{align}
for all $x \in \lambda'^{-1}(i)$.
The average extractable work from the state transition is:
\begin{align}
    \avg{W} = \sum_{x=1}^{N} r_{x} (\epsilon_{x}' - \epsilon_{x})
\end{align}
and we have:
\begin{align}
\sum_{x=1}^{N} r_{x} \epsilon_{x} &=  \sum_{x=1}^{N} r_{x} \Bigg[ \epsilon_{1} + \kt \log \Bigg( \frac{r_{1}}{r_{x}}\frac{p_{i}'e^{\beta e_{i}}}{p_{1}'e^{\beta e_{1}}} \Bigg) \Bigg] \\
    &= \kt D_{1}(\boldsymbol{p}_{S}'||\boldsymbol{\tau}_{S}) +C \\
\sum_{x=1}^{N} r_{x} \epsilon'_{x} &=  \sum_{x=1}^{N} r_{x} \Bigg[ \epsilon_{1} + \kt \log \Bigg( \frac{r_{1}}{r_{x}}\frac{p_{i}e^{\beta e_{i}}}{p_{1}'e^{\beta e_{1}}} \Bigg) \Bigg] \\
    &= \kt D_{1}(\boldsymbol{p}_{S}||\boldsymbol{\tau}_{S}) +C
    ~,
\end{align}
where:
\begin{align}
C=\epsilon_{1}-e_{1} + \kt \Big(\sum_{x} r_x \log \frac{r_{1}}{r_{x} p'_{1}}- Z_{S}\Big)
\end{align}
is a constant. This recovers the stochastic thermodynamics result:
\begin{align}
\avg{W} =k_{B}T \left[D_{1}(\boldsymbol{p}_{S}||\boldsymbol{\tau}_{S})-D_{1}(\boldsymbol{p}'_{S}||\boldsymbol{\tau}_{S})\right]
  ~.
\end{align}

This gives the distribution $\{r_{i}\}_{i\in\mathcal{W}}$ and energy levels (Eqs. \eqref{eqn:initiallevels} and \eqref{eqn:finallevels}) for the efficient work reservoir explicitly, completing the construction.

\section{Constructing \texorpdfstring{$\lambda$}{lambda} and \texorpdfstring{$\lambda'$}{lambdaprime}} \label{appendix:lambdaconstruction}

This section constructs the mappings $\lambda$ and $\lambda'$ in Eqs. \eqref{initialworkresprob} and \eqref{finalworkresprob}. We have $p_{i}'=P_{i}'-P_{i-1}'$ and $P_{i}',\ P_{i-1}' \in \boldsymbol{R}$. We define sets $\sigma'_{i} \subseteq \{1,2,\cdots N\}$ such that:
\begin{align}
    \sum_{i\in \sigma'_{i}} r_{i} &= P_{i}'~.
\end{align}
We have $\sigma'_{0}=\{\}$, $\sigma'_{n}=\{1,2,\cdots N\}$, and $\sigma'_{0} \subset \sigma'_{1} \subset \cdots \subset \sigma'_{n}$. $\lambda:\{1,2,\cdots, N\} \to \{1,2,\cdots, n\}$ is defined by
$\lambda(\sigma'_{i}\setminus\sigma'_{i-1})=i$ for $i=\{1,2,\cdots,n\}$. We have:
\begin{align}
    \sum_{j\in \lambda^{-1}(i)} r_{j} &= \sum_{j\in \sigma'_{i}\setminus\sigma'_{i-1}} r_{j} \nonumber \\
    &=\sum_{j\in \sigma'_{i}} r_{j}-\sum_{j\in\sigma'_{i-1}} r_{j} \nonumber\\
    &=  P_{i}'-P_{i-1}' = p_{i}~.
\end{align}
We define $\lambda'$ similarly.

\section{A different way to construct efficient work reservoirs}
\label{appendix:diffworkreservoir}

This section presents an alternative construction of a work reservoir for trivial Hamiltonian $\boldsymbol{E}_{S}=0$. More directly, the efficient work reservoir for a transition is not unique.

Consider a $2n^2$-dimension work reservoir of which energy levels are $\boldsymbol{E}_{W} = \{\epsilon_{11},\cdots,\epsilon_{nn},\epsilon_{11}',\cdots\epsilon_{nn}'\}$. The initial work reservoir probability distribution is $(\boldsymbol{p}\otimes\boldsymbol{p}',\boldsymbol{0})$ and the final is $(\boldsymbol{0},\boldsymbol{p}\otimes\boldsymbol{p}')$. The energy levels satisfy:
\begin{align}
p_{1}p_{j}'e^{\beta \epsilon_{1j}}&=\cdots=p_{n}p_{j}'e^{\beta \epsilon_{nj}}=k_{j} \  \mathrm{for}\  j=1,\cdots,n \\ 
p_{i}p_{1}'e^{\beta \epsilon_{i1}'}&=\cdots=p_{i}p_{n}'e^{\beta \epsilon_{in}'}=k_{i}' \ \mathrm{for}\  i=1,\cdots,n \\
k_{i}'p_{j}'&=p_{i}k_{j} \  \mathrm{for \ any \ pair~}(i,j)
  ~.
\end{align}
These conditions ensure that the initial total curve coincides with the final curve. We have:
\begin{align}
    \epsilon_{ij}&= \kt \log\frac{k_{j}}{p_{i}p_{j}'} \\
    \epsilon_{ij}'&= \kt \log\frac{k_{i}'}{p_{i}p_{j}'}
    ~.
\end{align}
The amount of work that can be extracted is:
\begin{align}
    \avg{W}&= \sum_{ij} p_{i}p_{j}' (\epsilon_{ij}'-\epsilon_{ij}) \\
    &= \kt \sum_{ij} p_{i}p_{j}'\log\frac{k_{i}'}{k_{j}} \\
    &= \kt \sum_{ij} p_{i}p_{j}'\log\frac{p_{i}}{p_{j}'} \\
    &= \kt (H(\boldsymbol{p}')-H(\boldsymbol{p})).
\end{align}

\section{Efficient work reservoir examples}
\label{appendix:examplesofreservoir}

The following analyzes several efficient work reservoirs for nontrivial Hamiltonians and time-dependent Hamiltonians. 

We first study a nontrivial Hamiltonian. Consider a two-level system with the Gibbs distribution $\boldsymbol{\tau}_{S}=(e^{-\beta e_{1}}/Z_{S},e^{-\beta e_{2}}/Z_{S})=(\frac{2}{3},\frac{1}{3})$. We begin with the distribution $\boldsymbol{p}_{S}=(\frac{1}{2},\frac{1}{2})$ and end with $\boldsymbol{p}_{S}'=(\frac{1}{3},\frac{2}{3})$. For the efficient work reservoir, we set $p_{W}=(\boldsymbol{r},\boldsymbol{0})$ and $p_{W}'=(\boldsymbol{0},\boldsymbol{r})$, where $\boldsymbol{r}=(\frac{1}{2},\frac{1}{3},\frac{1}{6})$. The work reservoir energy levels satisfy $\exp(-\beta \epsilon_{i})=\left\{\frac{1}{4}a, \frac{2}{3}a, \frac{1}{12}a\right\}$ for $i=1,2,3$ and $\exp(-\beta \epsilon_{i}')=\left\{\frac{1}{3}a, \frac{4}{9}a, \frac{2}{9}a\right\}$, for $i=1,2,3$ and where $a$ is a positive number. The work reservoir's energy change is:
\begin{align}
    \avg{W}&= \sum_{i=1}^{3} r_{i} (\epsilon_{i}'-\epsilon_{i}) \\
    &= \frac{1}{2} \kt \log\frac{3}{4} + \frac{1}{3} \kt \log\frac{3}{2} + 
    \frac{1}{6} \kt \log\frac{3}{8} \\
    &= \kt D_{1}(\boldsymbol{p}_{S}||\boldsymbol{\tau}_{S}) - \kt D_{1}(\boldsymbol{p}_{S}'|| \boldsymbol{\tau}_{S}) \\
    &= -0.17216 ~ \kt~.
\end{align}
The amount of the work is the negative nonequilibrium free energy difference. Table \ref{tab:nontrivialHamiltonianex} plots the thermomajorization curves.

\begin{table*}[!t]
    \scalebox{0.25}{
   \includegraphics[]{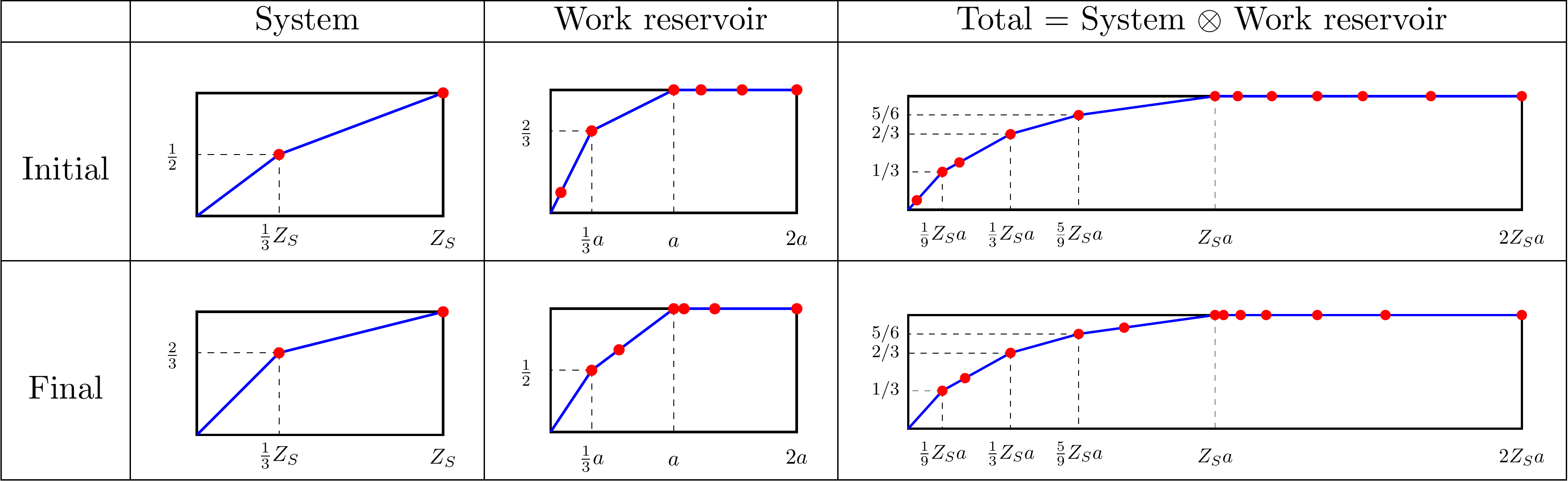}
   }
\caption{Thermomajorization curves with elbow point coordinates of $\rho_{S}$, $\rho_{W}$, $\rho_{SW}$, $\rho_{S}'$, $\rho_{W}'$, and $\rho'_{SW}$ for a state transition with a nontrivial Hamiltonian. For initial work reservoir, the red points $x-$axis coordinates are $\frac{1}{12}a$, $\frac{1}{3}a$, $a$, $\frac{11}{9}a$, $\frac{14}{9}a$ and $2a$, respectively. For initial total curves, the $x-$axis coordinates are $\frac{1}{36}Z_{S}a$, $\frac{1}{9}Z_{S}a$, $\frac{5}{36}Z_{S}a$, $\frac{1}{3}Z_{S}a$, $\frac{5}{9}Z_{S}a$, $Z_{S}a$, $\frac{29}{27}Z_{S}a$, $\frac{32}{27}Z_{S}a$, $\frac{36}{27}Z_{S}a$, $\frac{40}{27}Z_{S}a$, $\frac{46}{27}Z_{S}a$, $2Z_{S}a$, respectively. For final work reservoir, the red points $x-$axis coordinates are $\frac{1}{3}a$, $\frac{5}{9}a$, $a$, $\frac{13}{12}a$, $\frac{4}{3}a$ and $2a$, respectively. For final total curve, the red points $x-$axis coordinates are $\frac{1}{9}Z_{S}a$, $\frac{5}{27}Z_{S}a$, $\frac{1}{3}Z_{S}a$, $\frac{5}{9}Z_{S}a$, $\frac{19}{27}Z_{S}a$, $Z_{S}a$, $\frac{37}{36}Z_{S}a$, $\frac{39}{36}Z_{S}a$, $\frac{42}{27}Z_{S}a$, $\frac{48}{27}Z_{S}a$, $\frac{56}{27}Z_{S}a$, $2Z_{S}a$, respectively.}
\label{tab:nontrivialHamiltonianex}
\end{table*}

The second example concerns a state transition under a time-dependent Hamiltonian. The initial distribution is $\boldsymbol{p}_{S}= (\frac{1}{2},\frac{1}{2})$ and the initial Gibbs distribution is $\boldsymbol{\tau}_{S}= (e^{-\beta e_{1}}/Z_{S}, e^{-\beta e_{2}}/Z_{S})= (\frac{1}{3}, \frac{2}{3})$. The final distribution is $\boldsymbol{p}_{S}'= (\frac{2}{3},\frac{1}{3})$ and final Gibbs distribution is $\boldsymbol{\tau}_{S}'= (e^{-\beta e_{1}'}/Z_{S}', e^{-\beta e_{2}'}/Z_{S}')= (\frac{1}{2}, \frac{1}{2})$. For the efficient work reservoir, we set $\boldsymbol{p}_{W}=(\boldsymbol{r},\boldsymbol{0})$ and $\boldsymbol{p}_{W}'=(\boldsymbol{0},\boldsymbol{r})$ where $\boldsymbol{r}=(\frac{1}{2},\frac{1}{3},\frac{1}{6})$. The work reservoir energy levels satisfy $\exp(-\beta \epsilon_{i})=\left\{\frac{3}{8}a, \frac{1}{2}a, \frac{1}{8}a\right\}$, for $i=1,2,3$ and $\exp(-\beta \epsilon_{i}')=\left\{\frac{1}{3} \frac{Z_{S}}{Z_{S}'}a, \frac{4}{9}\frac{Z_{S}}{Z_{S}'}a, \frac{2}{9}\frac{Z_{S}}{Z_{S}'}a\right\}$ for $i=1,2,3$. The work reservoir's energy change is:
\begin{align}
    \avg{W}&= \sum_{i=1}^{3} r_{i} (\epsilon_{i}'-\epsilon_{i}) \\
    &= \frac{1}{2} \kt \log\frac{9}{8}\frac{Z_{S}'}{Z_{S}} + \frac{1}{3} \kt \log\frac{9}{8}\frac{Z_{S}'}{Z_{S}} \\
    & +\frac{1}{6} \kt \log\frac{9}{16} \frac{Z_{S}'}{Z_{S}}\\
    &= \kt (D_{1}(\boldsymbol{p}_{S}||\boldsymbol{\tau}_{S}) - \log Z_{S}) \\ 
    & -\kt( D_{1}(\boldsymbol{p}_{S}'|| \boldsymbol{\tau}_{S}) - \log Z_{S}') \\
    &= (0.0022585 +\log\frac{Z_{S}'}{Z_{S}} )\kt
    ~.
\end{align}
The amount of work is the nonequilibrium free energy difference. Table \ref{tab:timedepHamiltonianex} plots the thermomajorization curves.

\begin{table*}[!t]
    \scalebox{0.25}{
    \includegraphics[]{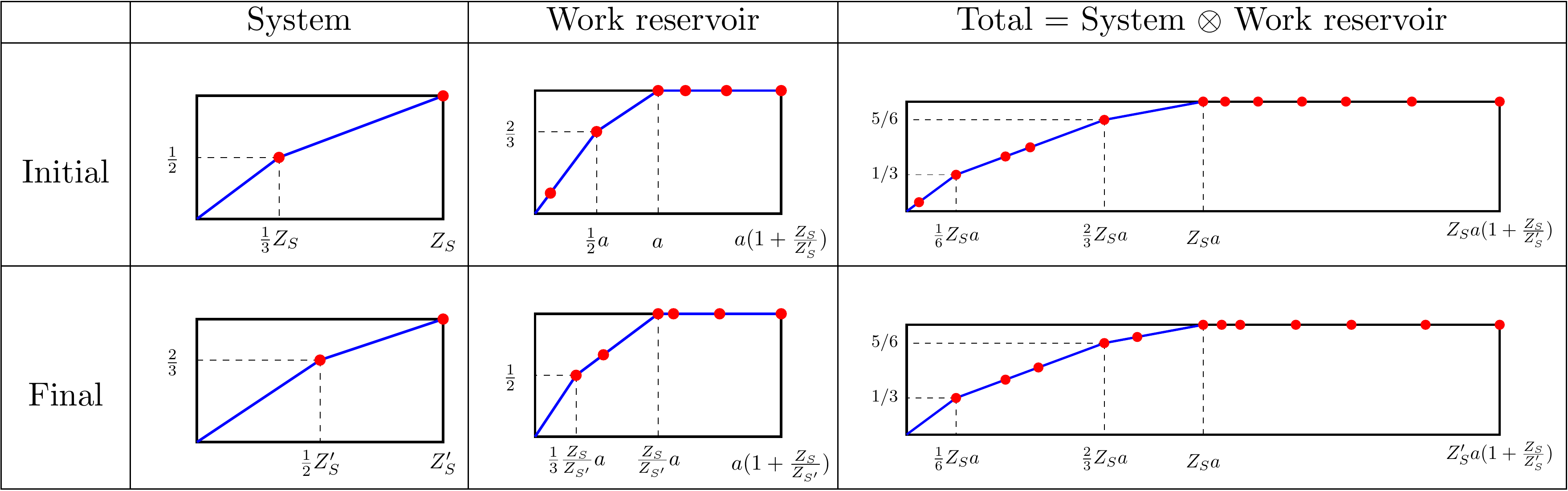}
    }
\caption{Thermomajorization curves with elbow point coordinates of $\rho_{S}$, $\rho_{W}$, $\rho_{SW}$, $\rho_{S}'$, $\rho_{W}'$, and $\rho'_{SW}$ for a state transition under a time-dependent Hamiltonian. For initial work reservoir, the red points $x-$axis coordinates are $\frac{1}{8}a$, $\frac{1}{2}a$, $a$, $(1+\frac{2}{9}\frac{Z_{S}}{Z_{S'}})a$, $(1+\frac{5}{9}\frac{Z_{S}}{Z_{S'}})a$ and $(1+\frac{Z_{S}}{Z_{S'}})a$, respectively. For initial total curves, the $x-$axis coordinates are $\frac{1}{24}Z_{S}a$, $\frac{1}{6}Z_{S}a$, $\frac{1}{3}Z_{S}a$, $\frac{5}{12}Z_{S}a$, $\frac{2}{3}Z_{S}a$, $Z_{S}a$, $Z_{S}a(1+\frac{2}{27}\frac{Z_{S}}{Z_{S}'})$, $Z_{S}a(1+\frac{5}{27}\frac{Z_{S}}{Z_{S}'})$, $Z_{S}a(1+\frac{9}{27}\frac{Z_{S}}{Z_{S}'})$, $Z_{S}a(1+\frac{13}{27}\frac{Z_{S}}{Z_{S}'})$, $Z_{S}a(1+\frac{19}{27}\frac{Z_{S}}{Z_{S}'})$, $Z_{S}a(1+\frac{Z_{S}}{Z_{S}'})$, respectively. For final work reservoir, the red points $x-$axis coordinates are $\frac{1}{3}\frac{Z_{S}}{Z_{S}'}a$, $\frac{5}{9}\frac{Z_{S}}{Z_{S}'}a$, $\frac{Z_{S}}{Z_{S}'}a$, $(\frac{1}{8}+ \frac{Z_{S}}{Z_{S}'})a$, $(\frac{1}{2}+ \frac{Z_{S}}{Z_{S}'})a$ and $(1+ \frac{Z_{S}}{Z_{S}'})a$, respectively. For final total curve, the red points $x-$axis coordinates are $\frac{1}{6}Z_{S}a$, $\frac{7}{18}Z_{S}a$, $\frac{1}{2}Z_{S}a$, $\frac{2}{3}Z_{S}a$, $\frac{7}{9}Z_{S}a$, $Z_{S}a$, $Z_{S}a(1+\frac{1}{16}\frac{Z_{S}'}{Z_{S}})$, $Z_{S}a(1+\frac{2}{16}\frac{Z_{S}'}{Z_{S}})$, $Z_{S}a(1+\frac{5}{16}\frac{Z_{S}'}{Z_{S}})$, $Z_{S}a(1+\frac{8}{16}\frac{Z_{S}'}{Z_{S}})$, $Z_{S}a(1+\frac{12}{16}\frac{Z_{S}'}{Z_{S}})$, $Z_{S}a(1+\frac{Z_{S}'}{Z_{S}})$, respectively.}
\label{tab:timedepHamiltonianex}
\end{table*}


\section{Thermomajorization curves form a monoid}
\label{appendix:monoid}

Abstract algebra defines a \emph{monoid} $M$ as a set equipped with an associative binary operation and an identity element. This appendix establishes that all possible thermomajorization curves at inverse temperature $\beta$ with the regular direct product form a monoid $M_{\beta}$.

For a thermomajorization curve $l$ with $n$ distinct slopes, we use a set with $n$ tuples to represent it:
\begin{align}
    \boldsymbol{l} = \{(y_{1},k_{1}),\cdots,(y_{n},k_{n})\}
    ~,
\end{align}
where $y_{i}$ and $k_{i}$ are the $y-$coordinate change and the slope of the $i-$th segment that satisfy $k_{1}>\cdots>k_{n}>0$ and $y_{1}+\cdots +y_{n}=1$. (We neglect subscripting with $\rho$ and $H$.) Note that this definition is not one to one: For a thermomajorization curve $l$, there may be many states corresponding to curve $l$. This appendix uses the thermomajorization curve $l$ and its representation $\boldsymbol{l}$ interchangeably.

The binary operation is defined as:
\begin{align}
\boldsymbol{l} \otimes \boldsymbol{m} := \left\{(y_{i}^{l}y_{j}^{m}, k_{i}^{l}k_{j}^{m}) \right\}_{i,j}/\sim
    ~,
\end{align}
where $\boldsymbol{l}=\{(y_{i}^{l}, k_{i}^{l})\}_{i}$, $\boldsymbol{m}=\{(y_{i}^{m}, k_{i}^{m})\}_{i}$, and $\sim$ means the segments with the same slopes are combined. The identity element is $\boldsymbol{I}=\{(1,1)\}$.

Verifying that the set of all thermomajorization curves forms a monoid $M_{\beta}$ is straightforward. In addition, $M_{\beta}$ is commutative; i.e., $\boldsymbol{l} \otimes \boldsymbol{m}=\boldsymbol{m} \otimes \boldsymbol{l}, \text{~for~all~} \boldsymbol{l}, \boldsymbol{m}\in M_{\beta}$. Not all elements in $M_{\beta}$ have corresponding inverses. Only the elements with the form $\{(a,1)\}$ have an inverse $\{(a^{-1},1)\}$. Thus, $M_{\beta}$ is a monoid and not a group. Although the inverse may not exist, we have the following theorem.

\begin{theorem}[Cancellative]
\label{thm:cancellativemonoid}
 If $\boldsymbol{x}, \boldsymbol{y}, \boldsymbol{a} \in M_{\beta}$ and $\boldsymbol{a} \otimes\boldsymbol{x}=\boldsymbol{a}\otimes\boldsymbol{y}$, then $\boldsymbol{x} = \boldsymbol{y}$.
\end{theorem}

\begin{proof}
The first element in $\boldsymbol{a}\otimes\boldsymbol{x}$ is $(y_{1}^{a} y_{1}^{x}, k_{1}^{a} k_{1}^{x})$ and the first element in $\boldsymbol{a}\otimes\boldsymbol{y}$ is $(y_{1}^{a} y_{1}^{y}, k_{1}^{a} k_{1}^{y})$. So, we have $y_{1}^{x}=y_{1}^{y}$ and $k_{1}^{x}=k_{1}^{y}$. Since we have:
\begin{align}
\boldsymbol{a} \otimes \boldsymbol{x} & = \boldsymbol{a} \otimes \boldsymbol{y} ~\text{and} \\
    \boldsymbol{a} \otimes (\boldsymbol{x}\setminus(y_{1}^{x},k_{1}^{x}))
    & = \boldsymbol{a} \otimes (\boldsymbol{y} \setminus (y_{1}^{y},k_{1}^{y}))
    ~,
\end{align}
we remove the same element on both sides. If we check the first element on both sides of the new equality, we have $y_{2}^{x}=y_{2}^{y}$ and $k_{2}^{x}=k_{2}^{y}$. Continuing this procedure, $y_{i}^{x} = y_{i}^{y}$ and $k_{i}^{x} = k_{i}^{y}$ for any $i$. Then we have $\boldsymbol{x}=\boldsymbol{y}$.
\end{proof}

These elementary facts allow exploring all possible work reservoirs for nondissipative state formation and work extraction. For state formation $(\tau,H)\to(\rho,H)$ with zero dissipation, we know the minimum segments of work reservoir's thermomajorization curve equal to the segments of $\rho$'s thermomajorization curve.

Suppose the corresponding initial work reservoir's thermomajorization curve is ${x}_{1}$. The final work reservoir's thermomajorization curve ${y}_{1}$ has only one segment. Thus, $\boldsymbol{y}_{1}$ has inverse $\boldsymbol{y}_{1}^{-1}$. Since there is no dissipation:
\begin{align}
\label{eqn:minimalworkreservoir}
    \boldsymbol{x}_{1} \otimes \boldsymbol{f}_{\tau,H } = \boldsymbol{y}_{1} \otimes \boldsymbol{f}_{\rho,H}
    ~.
\end{align}
Suppose there is another work reservoir suited for state formation whose initial and final thermomajorization curves are $\boldsymbol{f}_{{\rho}_{W},H_{W}}$ and $\boldsymbol{f}_{{\rho}_{W}',H_{W}}$. Then:
\begin{align}
\label{eqn:newworkreservoir}
    \boldsymbol{f}_{{\rho}_{W},H_{W}} \otimes \boldsymbol{f}_{\tau,H } = \boldsymbol{f}_{{\rho}_{W}',H_{W}} \otimes \boldsymbol{f}_{\rho,H}
    ~.
\end{align}
Multiply $\boldsymbol{x}_{1}$ on both sides of Eq. \eqref{eqn:newworkreservoir} and use Theorem \ref{thm:cancellativemonoid} to remove $\boldsymbol{f}_{\rho,H}$. Then:
\begin{align}
\boldsymbol{x}_{1}\otimes\boldsymbol{f}_{{\rho}_{W}',H_{W}}=\boldsymbol{y}_{1}\otimes\boldsymbol{f}_{{\rho}_{W},H_{W}}
~.
\end{align}
Since $\boldsymbol{y}_{1}$ has an inverse:
\begin{align}
\boldsymbol{f}_{{\rho}_{W},H_{W}} & = \boldsymbol{x}_{1}\otimes\boldsymbol{f}_{{\rho}_{W}',H_{W}} \otimes\boldsymbol{y}_{1}^{-1} \\
    & = \boldsymbol{b}\otimes\boldsymbol{x}_{1}
  ~,
\end{align}
where $\boldsymbol{b} = \boldsymbol{f}_{{\rho}_{W}',H_{W}} \otimes \boldsymbol{y}_{1}^{-1}$ or  $\boldsymbol{f}_{{\rho}_{W}',H_{W}} = \boldsymbol{b} \otimes \boldsymbol{y}_{1}$.

So, we write any general work reservoirs $\boldsymbol{f}_{{\rho}_{W},H_{W}}$ and $\boldsymbol{f}_{{\rho}_{W}',H_{W}}$ in terms of $\boldsymbol{x}_{1}$ and $\boldsymbol{y}_{1}$:
\begin{align}
\label{eqn:generalworkcurveofstateformation}
    \boldsymbol{f}_{{\rho}_{W},H_{W}} &= \boldsymbol{b} \otimes \boldsymbol{x}_{1} \nonumber \\
    \boldsymbol{f}_{{\rho}_{W}',H_{W}} &= \boldsymbol{b} \otimes \boldsymbol{y}_{1}
    ~.
\end{align}
This means the initial thermomajorization curve must be equal to the product of $\boldsymbol{x}_{1}$ and an arbitrary curve $\boldsymbol{b}$ and the final thermomajorization curve must equal the product of $\boldsymbol{y}_{1}$ and curve $\boldsymbol{b}$. These are the most general thermomajorization curves of the work reservoir for state formation with zero dissipation.

Next, we express this relation in terms of $\alpha-$R\'enyi divergences. Recall the definition of the $\alpha-$free energy of state $\rho$:
\begin{align}
F_{\alpha}(\rho) & = F_{eq}+\kt D_{\alpha}(\rho \| \tau) \\
    & = F_{eq} + \kt \frac{1}{\alpha -1}
    \log \left( \sum _{i=1}^{n}{\frac {p_{i}^{\alpha }}{q_{i}^{\alpha -1}}} \right)
  ~,
  \end{align}
where $\{p_{i}\}_{i=1}^{n}$ and $\{q_{i}\}_{i=1}^{n}$ are population vectors of state $\rho$ and Gibbs distribution and $F_{eq}=-\kt \log Z$ is the equilibrium free energy. The $\alpha-$free energy only depends on the thermomajorization curve's elbow points. Suppose $\rho$'s thermomajorization curve is $\boldsymbol{f}_{\rho,H}=\{(y_{i},k_{i})\}_{i=1}^{n}$, then:
\begin{align}
    D_{\alpha}(\rho \| \tau) &= \frac{1}{\alpha -1}\log {\Bigg (}\sum _{i=1}^{n}{\frac {p_{i}^{\alpha }}{(e^{-\beta \epsilon_{i}})^{\alpha -1}}} Z^{\alpha-1}{\Bigg )} \\
    &= \frac{1}{\alpha -1}\log {\Bigg (}\sum _{i=1}^{n}{\frac  {p_{i} p_{i}^{\alpha -1 }}{(e^{-\beta \epsilon_{i}})^{\alpha -1}}} Z^{\alpha-1}{\Bigg )} \\
    &= \frac{1}{\alpha -1}\log {\Bigg (}\sum _{i=1}^{n} {y_{i} k_{i}^{\alpha -1 }} Z^{\alpha-1}{\Bigg )}
    ~.
\end{align}
For any state $\rho$ and its thermomajorization curve $\boldsymbol{a}$, we use $F_{\alpha}(\boldsymbol{a})=F_{\alpha}(\rho)$ to denote the $\alpha-$free energy. For the general work curves $\boldsymbol{f}_{{\rho}_{W},H_{W}}$ and $\boldsymbol{f}_{{\rho}_{W}',H_{W}}$, from Eqs. \eqref{eqn:minimalworkreservoir} and \eqref{eqn:generalworkcurveofstateformation}, we have:
\begin{align}
    F_{\alpha}(\boldsymbol{x}_{1}) +F_{\alpha}(\tau) &= F_{\alpha}(\boldsymbol{y}_{1}) + F_{\alpha} (\rho) \\
    F_{\alpha}(\rho_{W}) &= F_{\alpha}(\boldsymbol{x}_{1}) + F_{\alpha}(\boldsymbol{b}) \\
    F_{\alpha}({\rho}_{W}') &= F_{\alpha}(\boldsymbol{y}_{1}) + F_{\alpha}(\boldsymbol{b})
    ~.
\end{align}
To remove $F_{\alpha}(\boldsymbol{b})$, we have:
\begin{align}
    \frac{e^{F_{\alpha}({\rho}'_{W})}}{e^{F_{\alpha}({\rho}_{W})}} = \frac{e^{F_{\alpha}(\boldsymbol{y}_{1})}}{e^{F_{\alpha}(\boldsymbol{x}_{1})}}=\frac{e^{F_{\alpha}(\tau)}}{e^{F_{\alpha}(\rho)}}
    ~,
\end{align}
where:
\begin{align}
   e^{F_{\alpha}(\tau)}& = e^{F_{eq}} \\
    e^{F_{\alpha}(\rho)}& = e^{F_{eq}}\cdot Z_{S}{\Big(}\sum_{i} p_{i} m_{i}^{\alpha -1}{\Big)}^\frac{1}{\alpha-1}
\end{align}
and $\boldsymbol{f}_{\rho, H}=\{(p_{i},m_{i})\}_{i}$. Then:
\begin{align}
    \frac{e^{F_{\alpha}({\rho}_{W}')}}{e^{F_{\alpha}({\rho}_{W})}} = \frac{{\Big(}\sum_{i} {p}_{i} {m}_{i}^{\alpha -1}{\Big)}^\frac{1}{1-\alpha}}{Z_{S}}
    .
\end{align}

This relation bridges between the work reservoir and the system and, thus, is a Jarzynski-like equality in the nondissipative scenario. Thus, from information about work we learn system transitions \cite{jarzynski1997nonequilibrium}. For a general nondissipative state transition, we cannot write the general work reservoir thermomajorization curves as in Eq. \eqref{eqn:generalworkcurveofstateformation}.

\begin{figure}[!t]
    \centering
    \scalebox{0.164}{
    \includegraphics[]{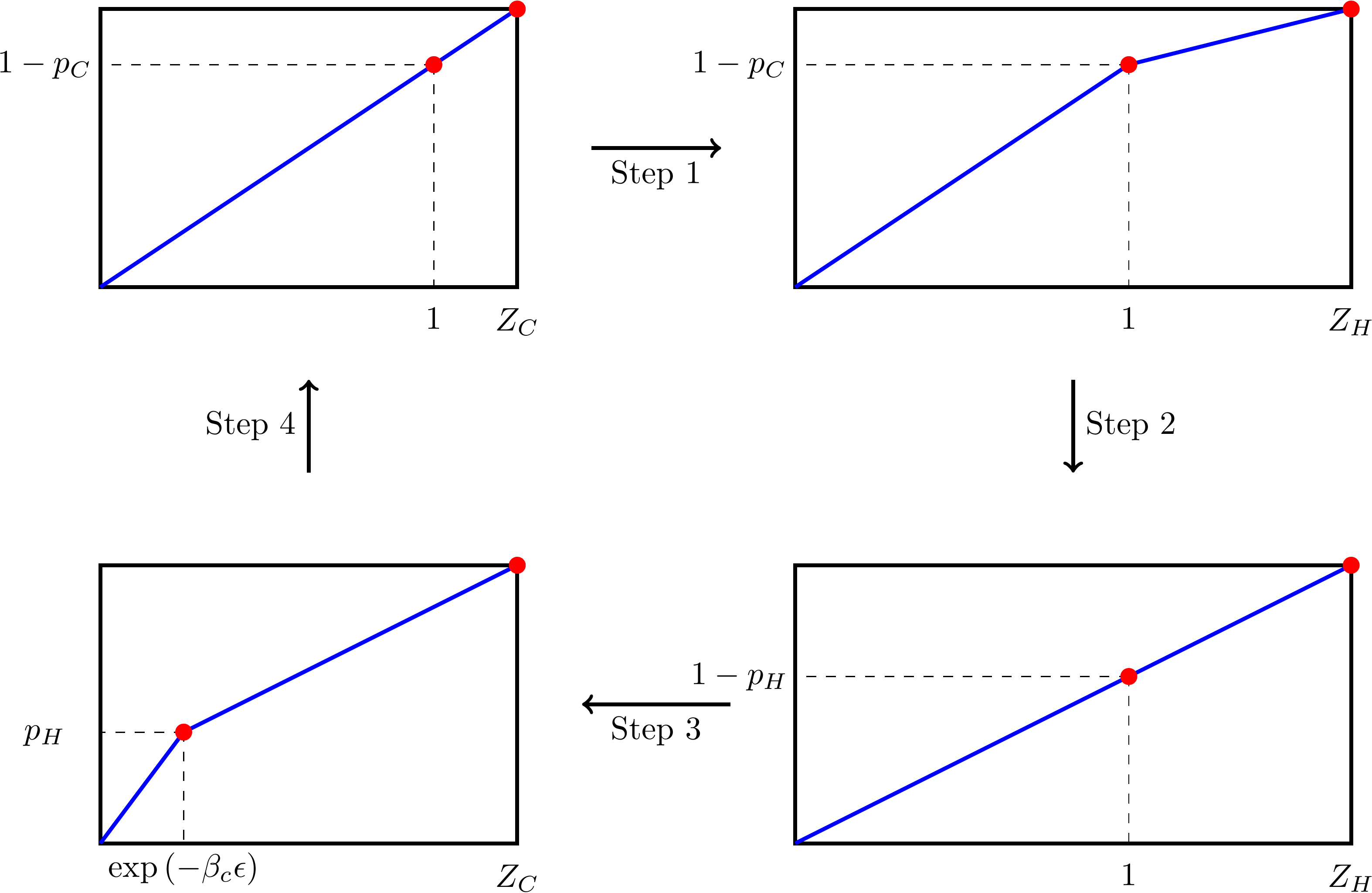}
    }
\caption{Thermomajorization curves for each stage of the qubit engine,  where $p_{C}=e^{-\beta_{C}\epsilon}/Z_{C}$ and $p_{H}=e^{-\beta_{H}\epsilon}/Z_{H}$. $Z_{C}$ and $Z_{H}$ are partition functions of the engine at temperature $T_{C}$ and $T_{H}$, respectively. }
\label{fig:Carnotcycle}
\end{figure}

\begin{table*}[t]
    \centering
    \scalebox{0.85}{
    \begin{tabular}{|c|c|c|c|}
    \hline 
        & & & \\
        & $\rho_{W1}$ & $\rho_{W2}$ & $\rho_{W3}$ \\
        & & & \\
    \hline
        & & & \\
    $p_{C}p_{H}$ & $-\frac{1}{\beta_{H}}\log(c_{1}p_{C})-\frac{1}{\beta_{C}}\log(c_{2}p_{H})$ & $-\frac{1}{\beta_{H}}\log(c_{1}p_{H})-\frac{1}{\beta_{C}}\log(c_{2}p_{H})$ & $-\frac{1}{\beta_{H}}\log(c_{1}p_{H})-\frac{1}{\beta_{C}}\log(c_{2}p_{C})$ \\
        & & & \\
    \hline
    & & & \\
    $p_{C}(1-p_{H})$ & $-\frac{1}{\beta_{H}}\log(c_{1}p_{C})-\frac{1}{\beta_{C}}\log(c_{2}(1-p_{H}))$ & $-\frac{1}{\beta_{H}}\log(c_{1}p_{H})-\frac{1}{\beta_{C}}\log(c_{2}(1-p_{H}))$ & $-\frac{1}{\beta_{H}}\log(c_{1}p_{H})-\frac{1}{\beta_{C}}\log(c_{2}(1-p_{C}))$\\
    & & & \\
    \hline
    & & & \\
    $(1-p_{C})p_{H}$ & $-\frac{1}{\beta_{H}}\log(c_{1}(1-p_{C}))-\frac{1}{\beta_{C}}\log(c_{2}p_{H})$ & $-\frac{1}{\beta_{H}}\log(c_{1}(1-p_{H}))-\frac{1}{\beta_{C}}\log(c_{2}p_{H})$ & $-\frac{1}{\beta_{H}}\log(c_{1}(1-p_{H}))-\frac{1}{\beta_{C}}\log(c_{2}p_{C})$\\
        & & & \\
    \hline
        & & & \\
    $(1-p_{C})(1-p_{H})$ & $-\frac{1}{\beta_{H}}\log(c_{1}(1-p_{C}))-\frac{1}{\beta_{C}}\log(c_{2}(1-p_{H}))$& $-\frac{1}{\beta_{H}}\log(c_{1}(1-p_{H}))-\frac{1}{\beta_{C}}\log(c_{2}(1-p_{H}))$ & $-\frac{1}{\beta_{H}}\log(c_{1}(1-p_{H}))-\frac{1}{\beta_{C}}\log(c_{2}(1-p_{C}))$ \\
        & & & \\
    \hline
    \end{tabular}
    }
\caption{Qubit engine efficient work reservoir energy levels: Here, we combine $W_{C}$ and $W_{H}$ into a single work reservoir. The work reservoir begins with $\rho_{W1}$. In step 2, the work reservoir changes from $\rho_{W1}$ to $\rho_{W2}$. And, in step 4, the work reservoir changes from $\rho_{W2}$ to $\rho_{W3}$. The nonzero components of probability distributions are $p_{C}p_{H}$, $p_{C}(1-p_{H})$, $(1-p_{C})p_{H}$, and $(1-p_{C})(1-p_{H})$. We list the corresponding energy levels in each work reservoir state, where $c_{1}$ and $c_{2}$ are two arbitrary positive constants.}
\label{tab:workconstructioncarnot}
\end{table*}

\begin{table}[h]
\scalebox{0.25}{
\includegraphics[]{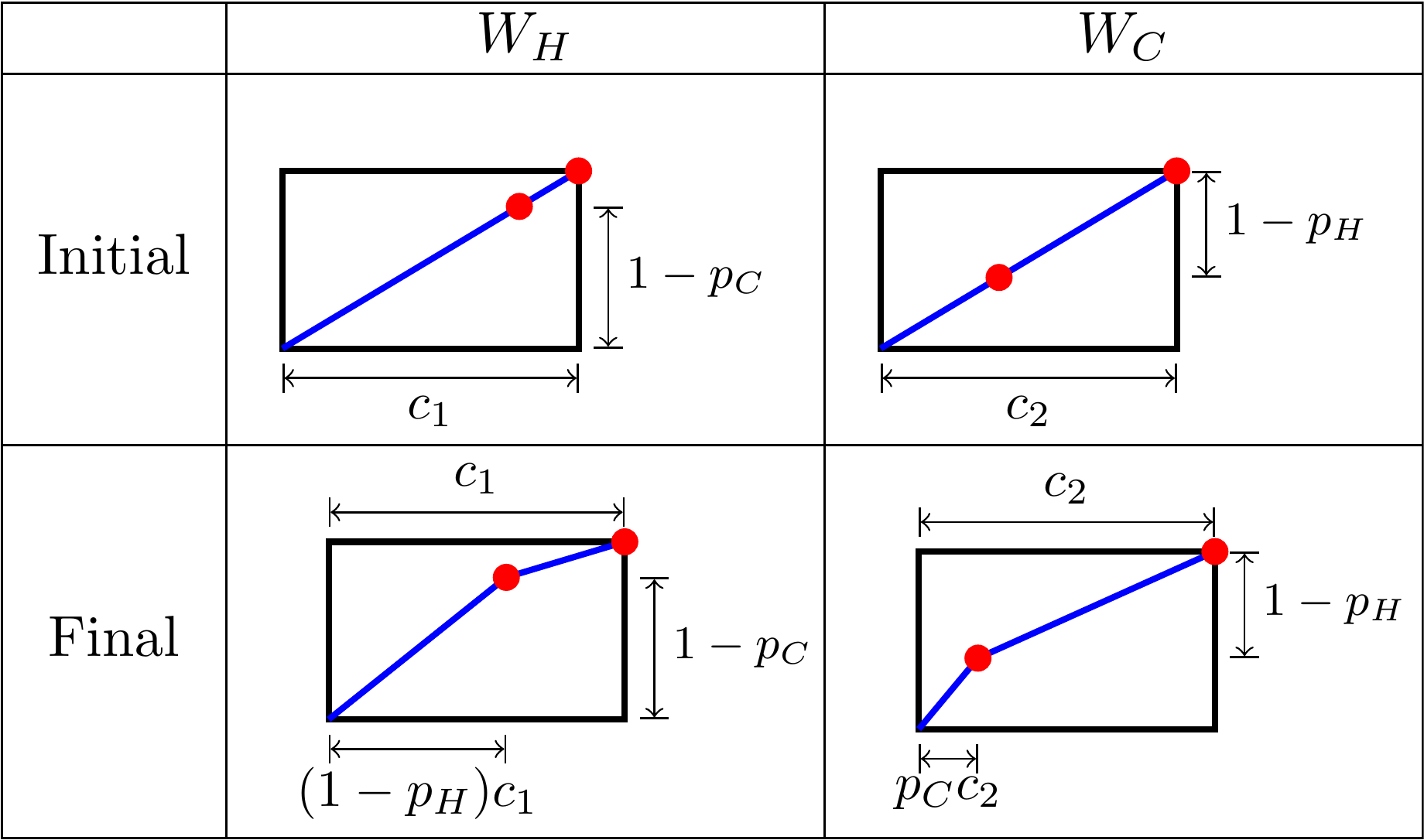}
}
\caption{Qubit engine thermomajorization curves of initial and final $W_{H}$ and $W_{C}$. $c_{1}$ and $c_{2}$ are arbitrary positive numbers. Here, we ignore flat portions in thermomajorization curves.}
\label{tab:workreservoirincarnotcycle}
\end{table}

\section{Carnot engines with efficient reservoirs}
\label{appendix:engine}

The following introduces a qubit engine implemented with efficient work reservoirs that executes a Carnot cycle. Note that when implemented with only two-level work reservoirs, the engine's efficiency is strictly vanishing \cite{woods2019maximum}.

In our setup, there are two thermal baths at temperatures $T_{C}$ and $T_{H}$ ($T_{C} <T_{H}$), two work reservoirs $W_{C}$ and $W_{H}$---that can be combined into one---and a system used as an engine. Since our engine and work reservoir run without dissipation, engine efficiency is $\eta=1-T_{C} / T_{H}$. The qubit engine's Hamiltonian is $H_{S}=\epsilon \ra{1}\la{1}$.

Initially, the engine is in thermal state $\tau_{C}$ at temperature $T_{C}$, being in contact with the cold bath. Next, $\tau_{C}$ is brought to the hot bath (Step 1) to extract work with work reservoir $W_{H}$ and ends in thermal state $\tau_{H}$ at temperature $T_{H}$ (Step 2). The work extracted from the hot thermal bath is $W_{H}=\kt_{H}D(\tau_{C}||\tau_{H})$. Then, the system returns to the cold bath (Step 3) and extracts work with reservoir $W_{C}$ and ends in thermal state $\tau_{C}$ at temperature $T_{C}$ (Step 4). The work that can be extracted from the cold thermal bath is $W_{C}=\kt_{C}D(\tau_{H}||\tau_{C})$. The cycle completes when the engine returns to the thermal state at $T_{C}$.

Now, let's construct the corresponding work reservoir for this Carnot cycle. The work reservoir's state only changes during steps 2 and step 4. Step 2 is a work extraction process. We use the minimal work reservoir $W_{1}$ to extract work without dissipation. In step 4, we also use the minimal work reservoir $W_{2}$ to extract work without dissipation. Figure \ref{fig:Carnotcycle} shows the system's thermomajorization curves for each step. And, Table \ref{tab:workreservoirincarnotcycle} shows the work reservoirs $W_{C}$ and $W_{H}$ used in the Carnot cycle. We can combine $W_{C}$ and $W_{H}$ into a single work reservoir. (See Table \ref{tab:workconstructioncarnot} for details.) Since there is no dissipation in steps 2 and step 4, the heat transferred to the hot bath $Q_{H}$ during step 2 and to the cold bath $Q_{C}$ during step 4 satisfy:
\begin{align}
    \beta_{H} Q_{H} + S(\tau_{H}) - S(\tau_{C}) &= 0 \\
    \beta_{C} Q_{C} + S(\tau_{C}) - S(\tau_{H}) &= 0
    ~.
\end{align}
Then we have $\beta_{C}Q_{C}+\beta_{H}Q_{H}=0$.

From energy conservation, the work done in one cycle is given by $W=-Q_{H}-Q_{C}$. And, the efficiency of this cycle is given by:
\begin{align}
    \eta & = \frac{W}{-Q_{H}} \\
    & = \frac{-Q_{H}-Q_{C}}{-Q_{H}} \\
    & =1-\frac{T_{C}}{T_{H}}
    ~.
\end{align}
For any engine operating with efficient work reservoirs, we always have:
\begin{align}
\beta_{C}Q_{C}+\beta_{H}Q_{H}=0
  ~.
\end{align}
As a result, the efficiency of an engine with efficient work reservoirs is always the Carnot efficiency $1-T_{C}/T_{H}$. Similar results are considered in Ref. \cite{bera2021attaining}.
\section{Realization of efficient state transitions}\label{appenix:finiterealization}
In this work, we study the possibility of realizing the state transitions by using multi-level work reservoirs. We may ask whether it is possible to construct a explicit joint unitary operator $U$ on the system, work reservoir plus the thermal bath to implement the transition. In this appendix, we give one example. 

We consider the famous Landauer's erasure beginning with the probability distribution $(\frac{1}{2},\frac{1}{2})$ storing in a trivial two-level system spanned by $\{ \ra{0},\ra{1}\}$ with Hamiltonian $H_{S}=0$. The work reservoir's is a two level system spanned by $\{ \ra{g},\ra{e}\}$ with Hamiltonian $H_{W}=\epsilon \ra{e}\la{e}$. Let us consider a special thermal bath with energy levels $\{0,1\epsilon,\cdots,N\epsilon\}$ of which corresponding degeneracy are $\{2^{0},2^{1},\cdots,2^{N}\}$, respectively. The Hamiltonian of the bath is 
\begin{align}
    H_{B}= \sum_{n=0}^{N}\sum_{i=1}^{2^{n}} n \epsilon \ra{n,i}\la{n,i}~
\end{align}
where $\ra{n,i}$ is the $i-$th degenerated eigenstate with the eigenvalue $n \epsilon$. The partition function of the bath is
\begin{align}
    Z_{B} = \frac{1-e^{-(N+1)\delta}}{1-e^{-\delta}}
\end{align}
where $\delta = \beta\epsilon - \log2$. 
We construct the joint energy preserving unitary $U$ such that
\begin{align}
    U \ra{0e} \otimes \ra{n,i} &= \ra{0g} \otimes \ra{n+1,i} \\
    U \ra{1e} \otimes \ra{n,i} &= \ra{0g} \otimes \ra{n+1,2^{n}+i} 
\end{align}
for $n={0,1,\cdots,N-1}$ and for $n=N$
\begin{align}
    U \ra{0e} \otimes \ra{N,i} &= \ra{0e} \otimes \ra{N,i} \\
    U \ra{1e} \otimes \ra{N,i} &= \ra{1e} \otimes \ra{N,i}~. 
\end{align}
There are undetermined degrees of freedom in this unitary operator $U$. But the conditions above are sufficient to determine our final state if the initial state is $\rho_{SWB}= (\frac{1}{2} \ra{0}\la{0}+\frac{1}{2} \ra{1}\la{1})\otimes \ra{e}\la{e} \otimes \tau_{B}$. The final state of the system and work reservoir is given by tracing out the thermal bath degrees of freedom
\begin{align}
    \rho'_{SW}= \mathrm{Tr}_{B}\big( U \rho_{SWB} U^{\dagger}\big)~.
\end{align}
This leads to the final state of system plus work reservoir
\begin{align}
    \rho'_{SW}= p\ra{0g}\la{0g}+\frac{1}{2}(1-p)\ra{0e}\la{0e}+\frac{1}{2}(1-p)\ra{1e}\la{1e}
\end{align}
where
\begin{align}
    p = \frac{1-e^{-N(\beta \epsilon - \log2)}}{1-e^{-(N+1)(\beta \epsilon - \log2)}}=\frac{1-e^{-N\delta}}{1-e^{-(N+1)\delta}}~.
\end{align}

The norm-1 distance between $\rho_{SW}'$ and desired final state $\ra{0g}\la{0g}$ is
\begin{align}
    ||\rho_{SW}'-\ra{0g}\la{0g}||_{1}= 2 (1-p)~.
\end{align}
The energy change in the work reservoir is 
\begin{align}
    W = -\kt(\delta+\log2)\frac{1-e^{-N\delta}}{1-e^{-(N+1)\delta}}~.
\end{align}
As long as $\delta=\beta \epsilon -\log2 >0$, the partition function of the bath is finite for any $N$. And the norm-1 distance can be arbitrarily small as $N\to\infty$. The corresponding energy change in the work reservoir can be arbitrarily close to $\kt \log2$.

\section{Correlated catalysts for trivial Hamiltonian}\label{Appendix:correlatedcatalyst}
In this section, we show an interesting result on correlated catalysts. The catalysts can be used to decrease the entropy productions in state transitions in trivial Hamiltonian. 

\begin{theorem}
    Consider a $m$-dimensional system with Hamiltonian $H=0$. Given a state $\rho_{S}$, a thermal operation $\mathcal{E}$ and $\mathcal{E}(\rho_{S})=\sigma_{S}$, then there exists a catalyst $\omega_{C}$ such that
    \begin{enumerate}
        \item $\mathcal{T}(\rho_{S} \otimes \omega_{C})=\sigma_{SC}$
        \item $\mathrm{Tr}_{C} \sigma_{SC}=\sigma_{S}$ and $\mathrm{Tr}_{S} \sigma_{SC}=\omega_{C}$
        \item The entropy production of $\rho_{S} \otimes \omega_{C} \to \mathcal{T}(\rho_{S} \otimes \omega_{C})$ is 0.
    \end{enumerate}
\end{theorem}

We use probability distribution to replace the density matrix and prove a theorem first.

\begin{theorem}\label{thm:catalystwithH=0}
    Let $p_{X}, q_{X}$ be distributions on a finite set $X$ and $\sigma_{Y}$ be probability distribution on a finite set $Y$ and $T$ be a doubly stochastic matrix  such that $T \cdot p_{X}=q_{X}$. Then there exists a distribution $q_{XY}$ on $X \times Y$ such that
    \begin{align}
        p_{X} \otimes \sigma_{Y} =q_{XY}~
    \end{align}
    and $q_{XY}$'s marginal distribution of $X$ is $q_{X}$ and the marginal distribution of $Y$ is $\sigma_{Y}$.
    Here, two probability distribution being equal to each other means that they are same up to a reorder.
\end{theorem}

\begin{proof}
    We prove this by directly constructing $q_{XY}$. From the Birkhoff–von Neumann theorem, any doubly stochastic matrix $T$ can be written as a convex combination of permutation matrices
    \begin{align}
        T = \sum_{i=1}^{\alpha} \theta_{i} P_{\pi_i}
    \end{align}
    where $\sum_{i}\theta_{i}=1$ and $P_{\pi_i}$ is the permutation matrix corresponding to permutation $\pi_{i}$. Without loss of generality, we can assume $\theta_{i}$ are rational, i.e., $\theta_{i}=m_{i}/N$ where $m_{i},N \in \mathbb{Z}$.
    Here, we take 
    \begin{align}
        \sigma_{Y}=\frac{1}{N}\underbrace{(1,\cdots,1)}_{ \#~ \text{of }1s=N}~.
    \end{align}
    And we introduce a $m\times N$ matrix $C$ to express $p_{X}\otimes \sigma_{Y}$
    \begin{align}
        p_{X}\otimes \sigma_{Y}=C=\begin{pmatrix}
            \frac{p_{X1}}{N} & \cdots &\frac{p_{X1}}{N} \\
            \vdots & \vdots &\vdots \\
            \frac{p_{Xm}}{N} & \cdots &\frac{p_{Xm}}{N} \\
        \end{pmatrix}~.
    \end{align}
    In this matrix, if we add all components in each column together, we have probability distribution $\sigma_{Y}$ and if we add all components in each row together, we have probability distribution $p_{X}$. We know 
    \begin{align}
        q_{X} &= T\cdot p_{X} = \sum_{i} \theta_{i} P_{\pi_i}\cdot p_{X} = \sum_{i} \frac{m_{i}}{N} P_{\pi_i}\cdot p_{X}\nonumber \\
        &=\sum_{i} \frac{m_{i}}{N}  \pi_{i}(p_{X})~
    \end{align}
    where $\pi_{i}(p_{X})$ is the probability distribution after the permutation $\pi_{i}$ taking effect on $p_{X}$. Now we permute components in each column of $C$ to get $C'$
    \begin{align}
        C'=\begin{pmatrix}
            \pi_{1}(p_{X})/N &  \cdots &\pi_{\alpha}(p_{X})/N
        \end{pmatrix}
    \end{align}
    where for each $\pi_{i}(p_{X})$ there are $m_{i}$ copies in $C'$. We let $q_{XY}$ be $C'$. Since we only permute components in each column, so we have same $\sigma_{Y}$ after we trace out $X$ and we have $q_{X}$ after we trace out $Y$. We only reorder the components in the matrix so we have
    \begin{align}
        p_{X}\otimes \sigma_{Y} =q_{XY}~.
    \end{align}
\end{proof}

From the theorem proved above, we see that for any state $\rho$, $\sigma$ and let $\mathcal{E}$ be a thermal operation, such that $\mathcal{E}(\rho)=\sigma$. We can find a catalyst $\omega$ such that the entropy production of the transition $\rho\otimes\omega \to \sigma' $ through a thermal operation is 0 and $\omega = \mathrm{tr}_1[\sigma']$, $\sigma = \mathrm{tr}_2[\sigma']$ if the Hamiltonian of the system is trivial. From Thm. \ref{thm:catalystwithH=0}, we can achieve this by only reordering the eigenvalues of $\rho \otimes \omega$. So we can even only use a unitary operator $U$ to achieve this, i.e., $U\rho \otimes \omega U^{\dagger}=\sigma'$. For nontrivial Hamiltonian in general, this is not correct. We cannot do arbitrary permutations since it violates energy conservation.

\clearpage

\bibliography{ref}

\end{document}